\newcommand{\mytitle}{Placing Green Bridges Optimally,
with~Close-Range~Habitats~in~Sparse~Graphs}

\def\thxPACSs{Funded by Deutsche Forschungsgemeinschaft (DFG, German
Research Foundation), project PACS (FL~1247/1-1, 522475669).}
\def\HUaffil{Humboldt-Universität zu Berlin,
Department of Computer Science, Algorithm Engineering Group, Germany}
\def\TUaffil{Technische Universität Berlin, Algorithmics and Computational Complexity, Germany}

\documentclass[twocolumn,pagebackref]{article}

\makeatletter
\renewenvironment{abstract}{%
	\small
	\quotation
	\noindent{\bfseries\abstractname.}%
}{\endquotation}
\renewcommand\section{\@startsection {section}{1}{\z@}%
																	{-2.5ex \@plus -1ex \@minus -.2ex}%
																	{1.3ex \@plus.2ex}%
																	{\normalfont\large\bfseries}}
\renewcommand\subsection{\@startsection{subsection}{2}{\z@}%
																		{-2.25ex\@plus -1ex \@minus -.2ex}%
																		{0.5ex \@plus .2ex}%
																		{\normalfont\normalsize\bfseries}}
\renewcommand\subsubsection{\@startsection{subsubsection}{3}{\z@}%
																		{-2.25ex\@plus -1ex \@minus -.2ex}%
																		{0.5ex \@plus .2ex}%
																		{\normalfont\normalsize\bfseries}}
\renewcommand\paragraph{\@startsection{paragraph}{4}{\z@}%
																		{2.25ex \@plus1ex \@minus.2ex}%
																		{-1em}%
																		{\normalfont\normalsize\bfseries}}
\renewcommand\subparagraph{\@startsection{subparagraph}{5}{\parindent}%
																			{2.25ex \@plus1ex \@minus .2ex}%
																			{-1em}%
																			{\normalfont\normalsize\bfseries}}
\makeatother

\usepackage[
includefoot,
top=1.6cm,
bottom=1.3cm,
left=1.35cm,
right=1.35cm,
a4paper,
]{geometry}

\usepackage[square,numbers]{natbib}
\usepackage{soul}
\usepackage{url}
\usepackage[x11names, table]{xcolor}
\usepackage[hidelinks]{hyperref}
\usepackage[utf8]{inputenc}
\usepackage[small]{caption}
\usepackage{doi}
\usepackage[T1]{fontenc}
\usepackage{subcaption}
\usepackage{amsmath}

\usepackage[capitalize,nameinlink]{cleveref}
\usepackage{paralist}

\usepackage{graphicx}
\usepackage{booktabs}
\usepackage{multirow}

\usepackage[x11names, table]{xcolor}

\usepackage{amssymb}
\usepackage[outline]{contour}
\usepackage{dsfont}
\usepackage{tabularx}
\usepackage[ruled,vlined,linesnumbered]{algorithm2e}

\usepackage{tikz}
\usetikzlibrary{fit,calc,positioning,shapes}
\usetikzlibrary{decorations.markings}
\usetikzlibrary{decorations.pathreplacing}
\usepackage{pgfplots}
\pgfplotsset{compat=1.5}
\pgfplotsset{major grid style={very thin,gray!20!white}} %
\usepackage{pgfplotstable}
\usepackage{mathtools}

\usepackage{pifont}%

\usepackage{xargs}
\newcommandx{\set}[2][1=1]{\ensuremath{\{#1,\ldots,#2\}}}
\newcommandx{\mydefenv}[4][3=A]{%
  \newtheorem{#1}{#2}
  {#3}{A}{\crefname{#1}{#2}{#2s}}{\crefname{#1}{#2}{#3}}%
}

\usepackage{amsthm}
\usepackage{thmtools}
\theoremstyle{plain}

\newtheorem{theorem}{Theorem}
\newtheorem{lemma}{Lemma}

\newtheorem{observation}[lemma]{Observation}

\newtheorem{rrule}{Reduction Rule}
\theoremstyle{definition}

\newtheorem{definition}[lemma]{Definition}

\newtheorem{problem}{Problem}
\declaretheorem[style=definition,name=Construction,qed=$\diamond$]{construction}
\theoremstyle{remark}
\newtheorem{example}{Example}

\crefname{observation}{Observation}{Observations}
\crefname{lemma}{Lemma}{Lemmas}
\crefname{definition}{Definition}{Definitions}
\crefname{rrule}{Reduction Rule}{Reduction Rules}
\crefname{construction}{Construction}{Constructions}
\crefname{proposition}{Proposition}{Propositions}
\crefname{theorem}{Theorem}{Theorems}
\crefname{corollary}{Corollary}{Corollaries}
\crefname{line}{Line}{Lines}
\crefname{problem}{Problem}{Problems}
\crefname{figure}{Figure}{Figures}
\crefname{table}{Table}{Tables}
\crefname{subsection}{Section}{Sections}
\crefname{section}{Section}{Sections}
\crefname{algorithm}{Algorithm}{Algorithms}

\crefname{example}{Example}{Examples}

\Crefname{subsection}{Sec}{Sects}
\Crefname{section}{Sec}{Sects}
\Crefname{problem}{Prob}{Probs}
\Crefname{observation}{Obs}{Obs}
\Crefname{proposition}{Prop}{Props}
\Crefname{corollary}{Cor}{Cors}
\Crefname{theorem}{Thm.}{Thm.}

\newcommandx{\decprob}[6][3=Input,5=Question]{\begin{samepage}
  \begingroup
\begin{problem}\label{prob:#2}%
  {{\textsc{#1}}}%
  \nopagebreak[4]\end{problem}\nopagebreak[4]\vspace{-0.5em}%
  \par\noindent\hangindent=\parindent\textbf{#3}:  #4\nopagebreak[4]
  \par\noindent\hangindent=\parindent\textbf{#5}:  #6
  \par\medskip
  \endgroup
  \end{samepage}
}

\DeclareMathOperator{\dist}{dist}

\newcommand{\I}{\mathcal{I}}
\newcommand{\yes}{\textnormal{\texttt{yes}}}
\newcommand{\no}{\textnormal{\texttt{no}}}
\newcommand{\RD}{$(\Rightarrow)\quad$}
\newcommand{\LD}{$(\Leftarrow)\quad$}

\newcommand{\N}{\mathbb{N}}
\newcommand{\Nzero}{\N_0}
\renewcommand{\O}{\mathcal{O}}

\newcommand{\prob}[1]{\textnormal{\textsc{#1}}}
\newcommand{\PCVC}{\prob{Planar Cubic Vertex Cover}}
\newcommand{\gbp}{Green Bridges Placement}
\newcommand{\GBP}{\prob{GBP}}

\newcommandx{\DgbpTsc}[2][1=$d$]{\prob{#1-Diameter \gbp}}
\newcommandx{\DgbpAcr}[2][1=$d$]{\prob{#1-Diam \GBP}}
\newcommandx{\DwgbpTsc}[2][1=$d$]{\prob{#1-Diameter \gbp{} with Costs}}
\newcommandx{\DwgbpAcr}[2][1=$d$]{\prob{#1-Diam \GBP-C}}
\newcommandx{\GwgbpTsc}{\prob{General \gbp{} with Costs}}
\newcommandx{\GwgbpAcr}{\prob{Gen \GBP-C}}

\newcommandx{\TwoDiam}{\DgbpAcr[2]}
\newcommandx{\TwoDiamC}{\DwgbpAcr[2]}

\newcommand{\cocl}[1]{\ensuremath{\operatorname{#1}}}
\newcommand{\NP}{\cocl{NP}}
\newcommand{\nphard}{\NP-hard}

\newcommand{\calH}{\mathcal{H}}
\newcommand{\calP}{\mathcal{P}}

\newcommand{\CCC}{\mathcal{C}}
\newcommand{\FFF}{\mathcal{F}}
\newcommand{\HHH}{\mathcal{H}}
\newcommand{\III}{\mathcal{I}}
\newcommand{\R}{\mathbb{R}}

\newcommandx{\tref}[2][1=]{{\footnotesize [\Cref{#2}#1]}}

\newcommand{\xcase}[2]{\textit{Case~#1}: #2.}

\DeclareMathOperator{\diam}{diam}
\DeclareMathOperator{\opt}{opt}

\newcommand{\ceq}{\ensuremath{\coloneqq}}
\newcommand{\cif}{\text{if }}
\newcommand{\otw}{\text{otherwise}}

\newcommand{\kfoure}{\ensuremath{K_4{\rm -} e}}
\newcommand{\maxH}{\ensuremath{\eta}}

\newcommand{\setFc}{\ensuremath{F^*}}%

\newcommand{\allrrs}[1]{\cref{rr:infeas,rr:irredge,rr:irrvertex,rr:nontriangle,rr:no-triangle-and-cut,rr:forced-habs,rr:smallcomps,#1}}

\definecolor{lilla}{HTML}{750787}
\newcommand{\thecolor}{black}%

\usepackage{makecell}

\usepackage[text size=footnotesize,color=green!15!white,disable]{todonotes}

\usepackage{etoolbox}

\newcommand{\xcitet}[2]{\citet{#2}}

\newcommand{\ExternalLink}{%
\tikz[x=1.2ex, y=1.2ex, baseline=-0.05ex]{%
    \begin{scope}[x=1ex, y=1ex]
        \clip (-0.1,-0.1) --++ (-0, 1.2) --++ (0.6, 0) --++ (0, -0.6) --++ (0.6, 0) --++ (0, -1);
        \path[draw, line width = 0.5, rounded corners=0.5] (0,0) rectangle (1,1);
    \end{scope}
    \path[draw, line width = 0.5] (0.5, 0.5) -- (1, 1);
    \path[draw, line width = 0.5] (0.6, 1) -- (1, 1) -- (1, 0.6);
    }
}

\newcommand{\tikzpreamble}{%
  \tikzstyle{xnode}=[scale=0.33,circle,fill,draw,color=\thecolor]
  \tikzstyle{xedge}=[thick,-,color=\thecolor]
  \tikzstyle{xxedge}=[ultra thick,-,color=\thecolor]
  \tikzstyle{xedgedot}=[thick,-,dotted,color=white]
  \tikzstyle{xhabA}=[-,opacity=0.2, line width=5pt, line cap=round,color=magenta]
  \tikzstyle{xhabB}=[-,opacity=0.2, line width=5pt, line cap=round,color=green]
  \tikzstyle{xhabC}=[-,opacity=0.2, line width=5pt, line cap=round,color=cyan]
  \tikzstyle{xhabD}=[-,opacity=0.2, line width=5pt, line cap=round,color=blue]
  \tikzstyle{xhabE}=[-,opacity=0.2, line width=5pt, line cap=round,color=yellow]
  \tikzstyle{dockingedge}=[
			thick,
			dashed,
			green!70!black,
		]
	\tikzstyle{forcededge} = [
			thick,
			densely dotted,
			red!70!black,
		]

  \def\habcolA{magenta}
  \def\habcolC{cyan}
  \def\habcolD{blue}
  \def\habcolE{yellow}
  \tikzstyle{xstreet}=[-, line width=6pt, line cap=round,color=black!20!gray]
  \tikzstyle{xstreetpat}=[dashed,-,white,thick]
  \tikzstyle{xrailway}=[-,double,double distance=4pt,ultra thick,color=brown!50!red]
}

\tikzset{
	street/.style = {
		line width = 6pt,
		line cap = round,
		color = black!20!gray,
	},
	bigstreet/.style = {
		line width = 12pt,
		color = black!20!gray,
	},
	streetlines/.style = {
		dashed,
		white,
		thick,
	},
	streetlinessolid/.style = {
		white,
		thick,
	},
	railway/.style = {
		double,
		double distance=4pt,
		ultra thick,
		color=brown!50!red,
	},
	baseedge/.style = {
		semithick,
	},
	soledge/.style = {
		baseedge,
		ultra thick,
		color=red!50!black,
	},
}

\newcommandx{\tikzES}[2][1={-,thick}]{%
	\foreach \x/\y in {#2}{\draw[#1] (\x) to (\y);}
}
\newcommandx{\drawAE}[4][3={-,thick}]{
	\draw[#3] (#1) to node[midway,inner sep=1pt,fill=white,sloped,font=\footnotesize]{#4}(#2);
}

\newcommandx{\tikzESd}[2][1={-,thick}]{%
	\foreach \x/\y/\z in {#2}{\drawAE{\x}{\y}[#1]{\z};}
}

\newcommand{\bspframe}{
  \node at (-1*\xr,2*\yr)[]{};
  \node at (2*\xr,-2*\yr)[]{};
}
\newcommand{\bsphabs}{
  \node (fox) at (-0.55*\xr,-.9*\yr)[]{};
  \node (goose) at (1*\xr,0.5*\yr)[]{};
  \node (salam) at (1*\xr,1*\yr)[]{};
  \node (tort) at (0*\xr,0.5*\yr)[]{};
  \draw [draw=\habcolA!30!gray,very thick] (fox) ellipse (1.1*\xr cm and 0.85*\yr cm);
  \draw [draw=\habcolE!30!gray,very thick] (goose) ellipse (0.66*\xr cm and 1.25*\yr cm);
  \draw [draw=\habcolD!30!gray,very thick] (tort) ellipse (0.66*\xr cm and 1.25*\yr cm);
}
\newcommand{\bsphabsnames}{
  \node at (-0.65*\xr,-.55*\yr)[fill=white,inner sep=.5pt]{$H_1$};
  \node at (-0.34*\xr,0.5*\yr)[fill=white,inner sep=.5pt]{$H_2$};
  \node at (1.35*\xr,0.40*\yr)[fill=white,inner sep=.5pt]{$H_3$};
}
\newcommand{\bsptransp}{
  \draw[bigstreet] (0*\xr,0*\yr) to (1*\xr,0*\yr);
  \draw[streetlines] (0*\xr,0.133*\yr) to (1*\xr,0.133*\yr);
  \draw[streetlinessolid] (0*\xr,0.0*\yr) to (1*\xr,0.0*\yr);
  \draw[streetlines] (0*\xr,-0.133*\yr) to (1*\xr,-0.133*\yr);
  \foreach \x in {0}{
    \draw[street] (\x*\xr,-1*\yr) to (\x*\xr,2*\yr);
    \draw[streetlines] (\x*\xr,-.9*\yr) to (\x*\xr,2*\yr);
  }
  \foreach \x/\y in {1.5/-1}{
    \draw[street] (-1*\xr,\y*\yr) to (\x*\xr,\y*\yr);
    \draw[streetlines] (-1*\xr,\y*\yr) to (\x*\xr,\y*\yr);
  }
  \foreach \x/\y in {1.75/1}{
    \draw[street] (-0*\xr+\y*\xr,\y*\yr) to (\x*\xr,\y*\yr);
    \draw[streetlines] (-0*\xr+\y*\xr,\y*\yr) to (\x*\xr,\y*\yr);
  }
  \foreach \x in {-1}{
    \draw[street] (\x*\xr,-2*\yr) to (\x*\xr,2*\yr);
    \draw[streetlines] (\x*\xr,-2*\yr) to (\x*\xr,2*\yr);
  }
  \foreach \x in {1}{
    \draw[street] (\x*\xr,-2*\yr) to (\x*\xr,2*\yr);
    \draw[streetlines] (\x*\xr,-2*\yr) to (\x*\xr,2*\yr);
  }
}
\newcommand{\bspnodes}{

	\node (n1) at (-2*\xr,-1.5*\yr) [xnode] {};
	\node (n2) at (-1*\xr,-1.0*\yr) [xnode] {};
	\node (n3) at (-1*\xr,-2.0*\yr) [xnode] {};
	\node (n4) at (-0*\xr,-0.5*\yr) [xnode] {};
	\node (n5) at (-0*\xr,-1.5*\yr) [xnode] {};
	\node (n6) at ( 1*\xr,-0.0*\yr) [xnode] {};
	\node (n7) at ( 1*\xr,-1.0*\yr) [xnode] {};

}
\newcommand{\bspedges}{
	\draw[baseedge]
		(n1) to (n3) to (n5) to (n4) to (n6) to (n7)
		(n1) to (n2) to (n4) to (n7) to (n5) to (n2) to (n3);
}
\newcommand{\bspedgesX}{
	\draw[baseedge]
		(n1) to node[midway,above,sloped]{\scriptsize 1}(n3) to node[midway,above,sloped]{\scriptsize 1}(n5) to node[midway,above,sloped]{\scriptsize 3}(n4) to node[midway,above,sloped]{\scriptsize 1}(n6) to node[midway,above,sloped]{\scriptsize 1}(n7)
		(n1) to node[midway,above,sloped]{\scriptsize 1}(n2) to node[midway,above,sloped]{\scriptsize 1}(n4) to node[midway,above,sloped]{\scriptsize 1}(n7) to node[midway,above,sloped]{\scriptsize 1}(n5) to node[midway,above,sloped]{\scriptsize 1}(n2) to node[midway,above,sloped]{\scriptsize 1}(n3);
}
\newcommand{\bspsoledges}{
	\draw[soledge]
		(n1) to (n2) to (n3)
		(n2) to (n5) to (n4) to (n7)
		(n4) to (n6);
}
\newcommand{\bspsoledgesX}{
	\draw[soledge]
		(n1) to (n2)
		(n1) to (n3)
		(n2) to (n4)
		(n2) to (n5) to (n7) to (n6)
		(n4) to (n7);
	\draw[ultra thick, color=green!50!black] (n3) to (n5);
}
\newcommand{\bspnodehabs}{
  \draw[xhabA] (n1) to (n3) to (n5) to (n2) to (n1) (n2) to (n3);
  \draw[xhabD] (n2) to (n4) to (n5) to (n2);
  \draw[xhabE] (n4) to (n6) to (n7) to (n4) to (n5) to (n7);
}

\usepackage{authblk}
\title{\Large\bf \mytitle}
\author[1]{Christian Wallisch}
\author[2]{Till Fluschnik\footnote{\thxPACSs}}
\author[1]{Leon Kellerhals}

\affil[1]{\TUaffil}
\affil[2]{\HUaffil}

\date{}

\newcommand{\theabstract}{%
	We study a network design problem motivated by
	the challenge of placing wildlife crossings to reconnect fragmented habitats of animal species,
	which is among the 17 goals towards sustainable development by the UN:
	Given a graph, whose vertices represent the fragmented habitat areas and
	whose edges represent possible green bridge locations (with costs),
	and the habitable vertex set for each species' habitat,
	the goal is to find the cheapest set of edges such that each species' habitat is sufficiently connected.
	We focus on the established variant where a habitat is considered sufficiently connected if it has diameter two in the solution
	and study its complexity in cases justified by our setting
	namely small habitat sizes on planar graphs and graphs of small maximum degree~$\Delta$.
	We provide efficient algorithms and NP-hardness results for different values of $\Delta$ and maximum habitat sizes on general and planar graphs.
}

\begin{document}

\maketitle

\begin{abstract}
	\theabstract{}
\end{abstract}

\pagestyle{plain}

\section{Introduction}
\hyphenation{Flusch-nik}

Habitat fragmentation due to human-made infrastructure like roads or train tracks leads to wildlife-vehicle collisions~\cite{Shilling21},
a threat to animals so severe that it has a significant impact on biodiversity~\cite{bennett2017effects}.
Installing wildlife crossings like bridges, tunnels, ropes, etc., to connect habitat patches of different animal species
is a cost-efficient and effective~\cite{HuijserMHKCSA08,HuijserDCAM09} measure against wildlife-vehicle collisions.
A~critical property in the design of wildlife corridors is the distance between habitat patches and the number of obstacles that need to be crossed between them.
It seems desirable to place green bridges such that every animal can reach any part of its habitat by traversing only a few green bridges.
A natural question is where to build green bridges to achieve dense connection of habitats while keeping financial costs low.

\xcitet{Fluschnik and Kellerhals}{FluschnikK24} developed a framework to tackle such questions.
They use a graph in which the vertices represent patches of land, and edges correspond to places where wildlife crossings can be installed.
The \emph{habitats} of animals then correspond to vertex subsets.
\xcitet{Fluschnik and Kellerhals}{FluschnikK24} define multiple classes of problems that represent different demands of habitat connectivity.
In this work, we focus on one of them.
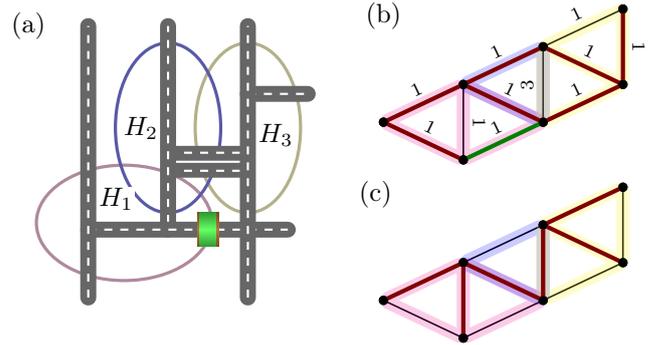
\begin{figure}
	\centering
	\begin{tikzpicture}
		\tikzpreamble{}
		\def\xr{1.05}
		\def\yr{0.9}

		\begin{scope}
		  \node at (-1.75*\xr,2*\yr)[]{(a)};
			\bspframe{}
			\bsphabs{}
			\bsphabsnames{}
			\bsptransp{}
			\def\hwMain{3pt}

			\def\hwSide{4pt}

			\definecolor{darkgreen}{RGB}{0,50,0}     %
			\definecolor{brightgreen}{RGB}{100,255,100} %

			\definecolor{darkbrown}{RGB}{70,30,0}
			\definecolor{brightbrown}{RGB}{180,120,60}

			\pgfdeclareverticalshading{greenGrad}{100bp}{
				color(0bp)=(darkgreen);
				color(50bp)=(brightgreen);
				color(100bp)=(darkgreen)
			}
			\pgfdeclareverticalshading{brownGrad}{100bp}{
				color(0bp)=(darkbrown);
				color(50bp)=(brightbrown);
				color(100bp)=(darkbrown)
			}

			\coordinate (A) at (0.5*\xr,-0.75*\yr);
			\coordinate (B) at (0.5*\xr,-1.25*\yr);

			\shade[shading=brownGrad]
				($(A)+(-\hwSide,0)$) rectangle ($(B)+(\hwSide,0)$);
			\shade[shading=greenGrad]
				($(A)+(-\hwMain,0)$) rectangle ($(B)+(\hwMain,0)$);

		\end{scope}
		\begin{scope}[xshift=4.7*\xr cm,yshift=-0.375*\yr cm]
			\node at (-2.1*\xr,-0.1*\yr)[]{(c)};
			\def\yr{1.0}
			\bspnodes{}
			\bspedges{}
			\bspnodehabs{}
			\bspnodes{}
			\bspsoledges{}
		\end{scope}
		\begin{scope}[xshift=4.7*\xr cm,yshift=2.25*\yr cm]
			\node at (-2/1*\xr,-0.1*\yr)[]{(b)};
			\def\yr{1.0}
			\bspnodes{}
			\bspedgesX{}
			\bspnodehabs{}
			\bspnodes{}
			\bspsoledgesX{}
		\end{scope}
	\end{tikzpicture}
	\caption{
		(a) A road network fragmenting habitats~$H_1$, $H_2$, and~$H_3$.
		(b \& c) The corresponding graph
		and a solution (thicker edges) for
		(b) \DwgbpAcr{} of minimum cost where the broader road part involves cost 3 (compared to 1 for the others) and the already built green bridge forms a forced edge;
		(c) \TwoDiam{} of minimum cardinality where edge costs and forced edges are ignored.
	}
	\label{fig:example}
\end{figure}

\decprob{\DwgbpTsc{} (\DwgbpAcr{})}{dwgbp}
{An undirected graph~$G=(V,E)$ with edge costs $c\colon E\to \Nzero$,
a set~$\calH=\{H_1,\dots,H_r\}$ of habitats with $H_i\subseteq V$ and~$|H_i|\geq 2$ for all~$i\in\set{r}$,
a set~$\setFc\subseteq E$ of \emph{forced} edges,
an integer~$k\in\Nzero$.}
{Is there a subset~$F\subseteq E$ with
$\setFc\subseteq F$ and
$c(F)\ceq \sum_{e\in F}c(e)\leq k$ such that
$\diam(G[F][H])\leq d$
for every~$H\in\HHH$?
}

If all edges have
\emph{unit} cost 1,
i.e.,
if $c\equiv 1$,
then we call the problem short \DgbpTsc{} (\DgbpAcr{}).
\xcitet{Fluschnik and Kellerhals}{FluschnikK24} study only \DgbpAcr{} and without forced edges.
Thus,
we generalize their model by allowing general edge costs.
The (possibly empty) set of forced edges that must be contained in every solution
only extends but does not generalize their model
since
any forced edge
can be modeled by a habitat of size two.
A forced edge models the case where a green bridge is already built or planned,
and thereby allows for incremental approaches of reconnecting habitats in an area.

\begin{example}
 \label{ex:intro}
 Consider the example given in~\cref{fig:example} for $d=2$.
 In~(a),
 we are given three habitats~$H_1$, $H_2$, and~$H_3$
 of different species
 fragmented by a street network.
 The middle horizontal street part is broader compared to all the others,
 and a green bridge is already built over a part of the horizontal street at the bottom.
 In (b)
 we show the graph~$G$ obtained from (a)
 with a vertex for each habitated land patch
 and with edge costs~$c$ (above each edge),
 where all edges have cost 1 except for the one that corresponds the broader street (which has edge cost 3).
 Each habitat~$H$ is depicted by a subgraph~$G[H]$ colored in correspondence with (a).
 The forced edge is drawn in green.
 Thick lines correspond to a solution of cost~$8$, which is minimum.
 In (c),
 we show the graph from~(b) when edge costs and forced edge are ignored.
 Thick lines correspond to a minimum-cardinality solution.
 If the true costs were applied and the forced edges would be added,
 the obtained solution would have cost~$9$.
 \hfill$\diamond$
\end{example}

This work contributes to a better understanding of \DgbpAcr{} from an algorithmic and complexity theoretic point of view.
As a starting point, and to garner better understanding of the structural properties of the problem, we focus solely on the case $d=2$.
This is also interesting from an ecological perspective as then animals from any two patches of the habitat can meet each other by each crossing at most one bridge.
Moreover, the case $d=2$ has also been studied the most in related problems (see below).
Our main interest is the complexity of \DwgbpAcr[2]{} on \emph{planar} graphs as well as grpahs with bounded \emph{maximum degree}.
As in our setting, the graphs correspond to the duals of road networks, these are natural assumptions.
Moreover, we focus on instances in which the habitats are small.
Small habitats appear more often for small mammals, amphibians, and reptiles
\cite{HammerM08,hennings2010wildlife}.

\paragraph{Related Work.}

The problems introduced by \xcitet{Fluschnik and Kellerhals}{FluschnikK24} are not the first
\emph{network design problems}
for \emph{wildlife preservation} studied from a theoretical and algorithmic perspective.
For designing wildlife corridors,
\xcitet{Lai et al.}{lai_steiner_2011}  consider the problem of finding a subgraph induced by a vertex set so that within each habitat, a set of habitat-specific terminals is connected.
\xcitet{LeBras et al.}{lebras_robust_2013} extend this problem to guarantee multiple disjoint paths between each terminal pair.
Both works highlight the importance of focusing on planar inputs.

The need to ``connect fragmented habitats'' also appears in a range of areas outside of wildlife preservation.
Applications include social networks \cite{angluin_inferring_2010}, graph drawing \cite{du1988interconnection,buchin2011planar,brandes_path_2012,klemz2014tree}, combinatorial auctions \cite{conitzer_combinatorial_2004}, vacuum technology \cite{du_complexity_1995}, structural biology \cite{agarwal_connectivity_2013}, and reconfigurable computer networks \cite{fan_algorithms_2008,chockler_constructing_2007}.
In the latter context, \xcitet{Chockler et al.}{chockler_constructing_2007} call for algorithms that ``connect habitats'' with a small diameter.
\xcitet{Chen et al.}{chen_polynomial_2015} list further applications of the studied problem and its variants.
Another problem concerned with connecting more than one habitat is \textsc{Steiner Forest} \cite{bateni_approximation_2011}, extending the well-known \textsc{Steiner Tree} problem.

\emph{Green bridges placement.}
\citet{FluschnikK24} proved that \TwoDiam{} is \nphard{} even if $r=1$,
and polynomial-time solvable if the graph has constant maximum degree and there are only constantly many habitats.
\xcitet{Herrendorf et al.}{HerrendorfKMS24} show that an algorithm solving the problem in $2^{o(n^2+r)}$ time would break the exponential time hypothesis, even if each habitat has size $4$.
Several related problems consider the task of connecting just one ``habitat''.
Minimizing the diameter of the habitat with a fixed budget is also \nphard{}~\cite{plesnik_complexity_1981},
as is the problem of adding few edges to a graph to make its diameter two~\cite{schoone_diameter_1987,frati_augmenting_2015}.

The variant where every habitat has size three has been studied extensively and is \nphard{} even if~$G$ is a clique~\cite{fan_algorithms_2008,fluschnik2022placing,hosoda_approximability_2012}.
It is $2$-approximable in polynomial time, which is best possible under the unique games conjecture~\cite{hosoda_approximability_2012}.
The variant is polynomial-time solvable on planar graphs if each habitat induces a triangle being the boundary of a face,
or when the graph has maximum degree and each habitat induces a triangle~\cite{fluschnik2022placing}.
If there is a solution that spans a tree in~$G$,
then it can be found in linear time~\cite{korach_clustering_2003}.

\TwoDiam{} has been studied by \xcitet{Jansson, Levcopoulos, and Lingas}{jansson_online_2021} who give a po\-ly\-no\-mi\-al-time $O(\maxH^4)$-approximation algorithm,
where~$\maxH$ denotes the maximum habitat size.
Moreover, there are several studies on heuristics for \TwoDiam{}~\cite{onus_minimum_2011, onus_parameterized_2016, oztoprak_topic_2017}.
Besides the primary objective, the heuristics aim to keep the maximum degree of $G[F]$ low.

Another related network design problem is network sparsification.
\xcitet{Gionis et al.}{gionis_community_2017} consider a variant of \TwoDiam{} called \prob{Sparse Stars} where an edge subset~$F$ is a solution if for each $H \in \mathcal{H}$ the graph $G[F][H]$ contains a spanning star.
As a star has diameter two, any such $F$ is also a solution for \TwoDiam{}.
They give a polynomial-time approximation algorithm for \prob{Sparse Stars}.
\xcitet{Herrendorf et al.}{HerrendorfKMS24} study the parameterized complexity of \prob{Sparse Stars} with respect to the parameters solution size $|F|$ and maximum degree~$\Delta$,
amongst others.
\xcitet{Korach and Stern}{korach_complete_2008} examine a variant of \prob{Sparse Stars} which additionally requires that the solution~$F$ induces a spanning tree of~$G$.

\paragraph{Our Contributions.}

Our results are summarized in~\cref{tab:results}.
\begin{table}[t!]
\caption{Overview on our results for \DgbpAcr[2]{} regarding maximum degree~$\Delta$, maximum habitat size, and further restrictions like planarity.
Recall that $n$, $m$, and $r$ denote the number of vertices, of edges, and of habitats, respectively.\\
($^\ast$: on planar graphs; $^\dagger$: even for unit costs)}
 \label{tab:results}
 \centering
	\begin{tabular}{@{}llll}\toprule
	 $\Delta$
 & \multicolumn{3}{c}{Maximum habitat size}
 \\ \cmidrule{2-4}
 &
 $\leq 3$ &
 $= 4$
 \qquad\quad
 &
 $\geq 5$
 \qquad\quad
 \\ \midrule
 $\leq 3$ & \multicolumn{3}{c}{\cellcolor{green!10!white} $\O(n+r)$ \tref{thm:maxdeg3}}
 \\
 $= 4$ & \multicolumn{2}{c}{\cellcolor{green!15!white}$\O(n+r)$ \tref{cor:deg4}} & \emph{open}
 \\
 $\geq 5$
 & \multicolumn{3}{c}{\cellcolor{red!10!white} \nphard$^\dagger$~\tref[(i)]{thm:nph:deg5}}
 \\
 $\geq 5$$^\ast$ & \cellcolor{green!10!white}
 $\O(n^2r^2+r^3)$~\tref{prop:planarH3}
 & \multicolumn{2}{c}{\cellcolor{red!15!white} \nphard$^\dagger$~\tref[(ii)]{thm:nph:deg5}}
 \\
 \bottomrule
 \end{tabular}
\end{table}
We provide tractability borders
(i.e., polynomial-time solvability versus \NP-hardness)
for \TwoDiamC{} when combining sparsity as described above with small maximum habitat sizes.
On the positive side,
we show that \DwgbpAcr[2]{} is polynomial-time solvable on instances with
\begin{compactenum}[(i)]
\item maximum degree three;
\item maximum degree four and maximum habitat size four;\label{item:deg4hab4}
\item planar graphs and maximum habitat size three.\label{item:planarH3}
\end{compactenum}
Cases (i) and (ii) are even solvable in linear time.
Complementing these,
we show that
\DgbpAcr[2]{} is \NP-hard even if the input graph
has maximum degree five and
\begin{compactenum}[(i)]
\item the maximum habitat size is three;
\item is also planar, and the maximum habitat size is four.
\end{compactenum}

For all of our polynomial-time algorithms, the habitat size is constant;
therefore the complexity of these specific cases arises from the way the habitats intersect.
For each of the polynomial-time solvable cases, we make use of different structural observations and algorithmic techniques.

Highlights of this work are the structural observations of po\-ly\-no\-mi\-al-time algorithms for instances with planar graphs and habitat size three (\cref{sec:planar})
and instances with maximum degree four and habitat size four (\cref{sec:deg4}).
The former algorithm exploits that size-three habitats in planar graphs admit a hierarchy which we can break down with an intricate data reduction rule.
For the latter result, we introduce the habitat intersection graph that captures the complexity of how habitats intersect; in this case, we are able to show that each connected component of the intersection graph is a path or a cycle, or has constant size.
For larger habitat sizes, we were unable to prove the same structure for the habitat intersection graph, but we conjecture that it keeps the same structure.
Both structural properties cease to hold
when increasing the habitat size to four in planar graphs,
or the maximum degree to five, even for size-three habitats.
This is reflected by our hardness results.

We first define some notation and state some preprocessing rules used throughout the paper in \cref{sec:prelim}.
Then,
we present the algorithms, starting with the algorithm on planar graphs in \cref{sec:planar}.
Afterwards, our results are ordered by increasing maximum degree (\cref{sec:deg3,sec:deg4}), ending on the hardness results for graphs of maximum degree five (\cref{sec:maxdeg5}).

\section{Preliminaries and Preprocessing}
\label{sec:prelim}

For an (undirected) graph $G = (V, E)$, we also denote by $V(G)$ and $E(G)$ the vertex and edge sets
and let $n \ceq |V|$ and $m \ceq |E|$.
For a vertex set~$W\subseteq V$,
the graph~$G[W]$ \emph{induced on~$W$} is the graph with vertex set~$W$ and all edges in~$E$ with both endpoints in~$W$.
For an edge set~$E'\subseteq E$,
the graph \emph{induced on~$E'$} is~$G[E'] \ceq (V, E')$.
For an edge set~$E'$ and a vertex set~$W\subseteq V$,
the graph~$G[E'][W]$ is the graph~$G[E']$ induced on~$W$.
For $u, v \in V$, let $\dist_G(u, v)$ be the length of a shortest path between $u$ and $v$.
The \emph{diameter} of $G$ is $\diam(G) \coloneqq \max_{u,v\in V} \dist_G(u,v)$.
The degree $\deg_G(v)$ of a vertex~$v$ is the number of edges with~$v$ as an endpoint.
The maximum degree of~$G$ is the maximum degree of a vertex in~$G$.
A graph is planar if it can be drawn into the two-dimensional plane without crossing edges.
A set~$W\subseteq V$ is a vertex cover of~$G$ if~$G[V\setminus W]$ is edgeless.

\paragraph{Preprocessing.} We use the following straightforward reduction rules throughout our paper.

\begin{rrule}
	\label{rr:infeas}
	Return \no{} if $\diam(G[H])>2$ for some $H \in \HHH$.
\end{rrule}

If there is an edge $e=\{u,v\}$ contained in a habitat $H$
such that the distance from $u$ to $v$ is more than two when removing the edge,
then $e$ must be in every solution.
Thus:

\begin{rrule}
	\label{rr:nontriangle}
	If there is a habitat~$H\in\HHH$ and an edge~$e\in E(G[H])$ not contained in a triangle in~$G[H]$,
	then mark~$e$ as forced.
\end{rrule}
\cref{rr:nontriangle} implies
the next more direct rule.
\begin{rrule}
	\label{rr:no-triangle-and-cut}
	If a habitat $H \in \HHH$ contains no triangle,
	then mark each edge in $G[H]$ as forced.
	If $G[H]$ for some $H \in \HHH$ contains a cut vertex~$v$,
	then mark all edges in~$G[H]$ incident to $v$ as forced.
\end{rrule}

Next, we remove all unnecessary habitats, edges, and vertices, and small components.

\begin{rrule}
	\label{rr:forced-habs}
	If $\diam(G[\setFc][H]) \le 2$ for a habitat $H \in \HHH$,
	then delete $H$ from $\HHH$.
\end{rrule}

\begin{rrule}
	\label{rr:irredge}
	If there is an edge~$e\in E$ such that there is no habitat~$H \in \HHH$ with~$e\in E(G[H])$,
	then delete~$e$. If further $e\in \setFc$, then also set~$k\ceq k-c(e)$.
\end{rrule}

\begin{rrule}
	\label{rr:irrvertex}
	If $v\in V$ is isolated,
	then delete~$v$.
\end{rrule}

\begin{rrule}
	\label{rr:smallcomps}
	Let~$C$ be a connected component of~$G$ with $|V(C)|\leq 6$.
	Then,
	in constant time,
	brute-force a minimum-cost set~$F_C^* \subseteq E(C)$
	with $\setFc \cap E(C) \subseteq F_C^*$
	such that $\diam(G[F_C^*][H])\leq 2$ for every habitat~$H\subseteq V(C)$.
	Next, delete~$C$ from~$G$ and
	all habitats~$H\subseteq V(C)$ from~$\calH$,
	and set~$k\ceq k-c(F_C^*)$.
\end{rrule}

\begin{observation}
	\label{obs:rr-applied}
	\allrrs{} can be exhaustively applied to an instance $(G, c, \HHH, \setFc, k)$ of \TwoDiamC{} in $\O(r \cdot \maxH^3 + n + m)$ time, where $\maxH \coloneqq \max_{H \in \HHH} |H|$.
	Each component has more than $6$ vertices and for each habitat $H$, $G[H]$
	(i) is $2$-connected,
	(ii) has at least two unforced edges and
	(iii) at least one triangle.
\end{observation}

\begin{proof}
	We first apply \cref{rr:infeas}; this takes $\O(r \cdot \maxH^3)$ time by computing at most $\maxH$ breadth-first searches for each habitat.
	Afterwards, we apply \cref{rr:nontriangle} (and thereby \cref{rr:no-triangle-and-cut}); this can be done alongside the breadth-first searches.
	Finally, we apply \cref{rr:irredge}, followed by \cref{rr:forced-habs} and \cref{rr:smallcomps} in linear time in the instance size.

 From~\cref{rr:smallcomps} it follows directly that $H$ is in a component with more than 6 vertices.

 (i) Since~$\diam(G[H])\leq 2$ due to~\cref{rr:infeas},
 we know that~$G[H]$ is connected.
 Suppose that~$G[H]$ is 1-connected.
 Then $H$ contains a cut vertex,
 and \cref{rr:forced-habs} or \ref{rr:no-triangle-and-cut} is applicable.

 (ii) Suppose that all but one of the edges in~$G[H]$ are forced.
 Since \cref{rr:infeas,rr:nontriangle} are inapplicable,
 we have $\diam(G[\setFc][H]) \le 2$,
 contradicting that \cref{rr:forced-habs} is inapplicable.

 (iii) If~$G[H]$ contains no triangle,
 then \cref{rr:no-triangle-and-cut} or \ref{rr:forced-habs} is applicable.
\end{proof}

We mention in passing that, after applying our preprocessing,
habiatats in instances with maximum degree $3$ ($4$) have size at most $6$ ($8$).
This can be verified by checking all graphs with maximum degree $4$ and diameter $2$, which have at most $15$ vertices \cite{elspas1964topological}.

\section{Planar graphs and habitats of size three}
\label{sec:planar}

\xcitet{Herkenrath et al.}{fluschnik2022placing} showed that, when every habitat has size three and induces a face in a planar graph, \TwoDiamC{} is efficiently solvable.\footnote{They show this for a variation of \TwoDiamC{} where one only needs that $G[F][H]$ is \emph{connected} for each $H \in \HHH$. For triangle habitats this is equivalent to our diameter constraint.}
We extend their result by showing that this is possible without the face constraint.

\begin{theorem}
	\label{prop:planarH3}
	\DwgbpAcr[2]{} on planar graphs is solvable in $O(n^2r^2+r^3)$ time
	if the habitat size is at most three.
\end{theorem}

\newcommand{\xin}{\mathrm{in}}
\newcommand{\xout}{\mathrm{out}}
\newcommand{\ain}[1]{\ensuremath{A^{\xin}_{#1}}}
\newcommand{\aout}[1]{\ensuremath{A^{\xout}_{#1}}}
\newcommand{\vin}[1]{\ensuremath{V^{\xin}_{#1}}}
\newcommand{\vout}[1]{\ensuremath{V^{\xout}_{#1}}}
\newcommand{\ein}[1]{\ensuremath{E^{\xin}_{#1}}}
\newcommand{\eout}[1]{\ensuremath{E^{\xout}_{#1}}}
\newcommand{\hin}[1]{\ensuremath{\HHH^{\xin}_{#1}}}
\newcommand{\hout}[1]{\ensuremath{\HHH^{\xout}_{#1}}}
\newcommand{\fin}[1]{\ensuremath{F_{#1}}}%
\newcommand{\fcin}[1]{\ensuremath{F^{*\xin}_{#1}}}

We assume that the graph is embedded in the plane $\R^2$ via a straight-line embedding $\phi$ and that each habitat induces a triangle by \cref{obs:rr-applied}.
The crucial observation
is that for any habitat $H \in \HHH$ of size 3, $G[H]$ splits the plane into an inner, bounded area $\ain{H} \subseteq \R^2$ and an outer, unbounded area $\aout{H} \subseteq \R^2$.
A vertex or edge is contained in (i) $G[H]$, (ii), the inner, or (iii) the outer area.
Habitats cannot be both inside and outside as then two edges would cross.
This implies a hierarchy on the habitats:
For every two habitats,
either one has a vertex in the inner area of the other,
or both habitats consider the other as outside of them.
Hence, there are habitats that induce faces (at the lowest level of the hierarchy), and thus, there is a habitat that contains only face habitats; we call such a habitat \emph{reducible}, see \cref{def:reducible}.

As habitats inside a reducible habitat $H$ may share edges with $H$,
the minimum cost of reconnecting them depends on the subsolution, i.e., the edges in $G[F][H]$.
The main tool for \cref{prop:planarH3} is \cref{rr:reducibleH}, which eliminates all vertices, edges, and habitats inside $H$,
propagating the costs of its optimal inside solutions into the cost of the edges of $H$.
This is remarkable since there are four possible subsolutions for $H$, but only three edges to store the information in:
A solution may contain all or all but one of the edges in $H$.

Exhaustive application of our reduction rule results in an instance in which every habitat induces a face and allows us to apply the algorithm by
\xcitet{Herkenrath et al.}{fluschnik2022placing}.

Let us introduce some notation.
For a habitat $H \in \HHH$,
a vertex $v \in V\setminus H$ is either inside or outside a habitat~$H$,
i.e., $\phi(v)\in \ain{H}$ or $\phi(v)\in \aout{H}$ (we define the vertices in $H$ to be neither inside or outside of $H$).
As $G$ is planar,
the endpoints of an edge in~$E\setminus E(G[H])$
cannot be both inside and outside of~$H$.
Hence,
$e\in E\setminus E(G[H])$ is inside if $e$ has an endpoint in~$\ain{H}$
and outside if $e$ has an endpoint in~$\aout{H}$.
Finally, we define a habitat $H' \in \HHH$, $H' \neq H$,
to be inside if $\ain{H'} \subsetneq \ain{H}$
and outside if $\aout{H'} \subsetneq \aout{H}$.
We denote by~$\vin{H},\vout{H},\ein{H},\eout{H},\hin{H}$, $\hout{H}$,
the sets of vertices, edges, and habitats that are inside or outside of~$H$.
Note that $H$ induces an (inner) face if and only if $\hin{H} = \emptyset$.

\newcommand{\reducible}{reducible}
\begin{definition}
	\label{def:reducible}
 We call a habitat~$H\in\HHH$ \emph{\reducible}
 if~$\hin{H}\neq \emptyset$ and for each~$H'\in\hin{H}$ it holds true that~$\hin{H'}= \emptyset$.
\end{definition}

We prove that if at least one habitat contains another,
then there is a \reducible{} habitat.

\begin{lemma}
 \label{lem:ExReducibleH}
 Let~$G$ be a non-empty finite planar graph with an embedding into the plane and~$\HHH\neq \emptyset$ a set of pairwise different habitats in~$G$ each inducing a triangle.
 Then either~$\hin{H}=\emptyset$ for all~$H\in\HHH$ or there is a \reducible{} habitat.
\end{lemma}

\begin{proof}
 Assume not.
 Then for every habitat~$H$ with~$\hin{H}\neq \emptyset$
 there is~$H'\in\hin{H}$ with~$\hin{H'}\neq \emptyset$.
 This way,
 we can construct a sequence~$H_1,\dots,H_j$ of habitats with~$j>\binom{|V(G)|}{3}$
 such that~$H_{i+1}\in\hin{H_{i}}$ for each~$i\in\set{j-1}$.
 Since~$G$ is finite and $j>\binom{|V(G)|}{3}$,
 there are at least two habitats~$H_q$ and~$H_p$ such that~$H_q=H_p$,
 a contradiction.
\end{proof}

We next define some terms that we need for \cref{rr:reducibleH}, alongside an example in \cref{fig:planarhab3}.
Let $H \in \HHH$ be a habitat;
define $E_H \coloneqq E(G[H])$.
The \emph{inner graph $G_H$ of $H$} is defined as the subgraph of $G$ containing all edges in $\ein{H} \cup E_H$.
The \emph{inner cost function} $c_H$ assigns cost $c_H(e) \coloneqq c(e)$ to each $e \in \ein{H}$ and cost $c_H(e) \coloneqq 0$ to each $e \in E_H$.
We define the \emph{inside optimum solution} $F_H$ for $H$ to be an edge set $E(G_H) \cap \setFc \subseteq F_H \subseteq E(G_H)$ satisfying $H' \subseteq V(G_H[F_H])$ and $\diam(G[F_H][H']) \le 2$ for all $H' \in \hin{H}$ that minimizes the cost $c_H(F_H) \eqqcolon \opt_H$.
Intuitively, $F_H$ is the cheapest solution with respect to the inner cost function that contains all forced edges satisfying the diameter constraints for all habitats inside of $H$.
In the example in \cref{fig:planarhab3}, $F_H = E_H \cup \{v_r,v_c\}$ and $\opt_H = 3$.

For $e \in E_H$, the \emph{$e$-omitting inner cost function} $c^e_H$ sets $c_H^e(e) \coloneqq 1\!+\!\sum_{e'' \in \ein{H}}c(e'')$ and $c^e_H(e') \coloneqq 0$ for each $e' \in E(G_H) \setminus \{e\}$.
The \emph{$e$-omitting inside optimum solution} $F^e_H$ and $\opt^e_H \coloneqq c^e_H(F^e_H)$ are defined just as the inside optimum solution with the only difference being the cost function.
In the example in \cref{fig:planarhab3}, $F^{\{v_r,v_b\}}_H = F_H \setminus \{\{v_r, v_b\}\} \cup \{\{v_\ell, v_c\}\}$ and $\opt^{\{v_r,v_b\}}_H = 11$.
Further, $\opt^{\{v_\ell,v_b\}}_H = 3$ and $\opt^{\{v_\ell,v_r\}}_H = 7$.

As the cost of $e$ is prohibitively high, we have
$e \notin F^e_H$.
Thus
$\smash{\opt_H\leq \opt^e_H \le \sum_{e' \in \ein{H}} c(e')}$,
and observe the following.

\begin{observation}
	\label{obs:inner-cost}
	For $H \in \HHH$, $e \in E_H$, and a solution $F$,
	\begin{equation*}
		c(F \cap \ein{H}) \ge
		\begin{cases}
			\opt_H, \text{ if } |F \cap E_H| = 3,\\
			\opt^e_H, \text{ if } E_H\setminus F = \{e\}.
		\end{cases}
	\end{equation*}
\end{observation}

This gives rise to our central reduction rule (see \cref{fig:planarhab3}).

\begin{figure}
 \centering
 \begin{tikzpicture}
  \def\xr{0.71}
  \def\yr{0.95}
  \def\xsh{4.}
  \def\teps{0.1}
  \tikzpreamble{}

  \newcommand{\PHTnodes}{%
		\node (l) at (-1.5*\xr,0*\yr)[xnode]{};
		\node (t) at (0*\xr,0.75*\yr)[xnode]{};
		\node (r) at (1.5*\xr,0*\yr)[xnode]{};
		\node (b) at (0*\xr,-2*\yr)[xnode]{};
		\tikzESd{l/t/1,r/t/3}
  }
  \newcommand{\PHTannot}[1]{%
		\node at (l)[label=90:{$v_\ell$}]{};
		\node at (t)[label=180:{$v_t$}]{};
		\node at (r)[label=90:{$v_r$}]{};
		\node at (b)[label=180:{$v_b$}]{};
		#1=0
			\node at (c)[label={[label distance=-6pt]225:{$v_c$}}]{};
		}
  \newcommand{\PHTnodesF}{%
    \PHTnodes{}
		\node (c) at (0*\xr,-0.75*\yr)[xnode]{};
		\tikzESd{l/c/4,r/c/3,b/c/8}
  }
  \newcommand{\labelk}[2]{
		\node at (-2.125*\xr,1.1*\yr)[anchor=west]{(#1) $k=#2$};
  }

  \newcommand{\PHTfixhab}{
		\draw[-, line width=5pt, opacity=0.2, color=magenta] ($(l)+(4*\teps,\teps)$)
		to ($(r)+(-4*\teps,\teps)$) to ($(t)+(0,-\teps)$) to cycle;
		\node at ($(t)-(0,0.4)$)[color=magenta]{$H'$};
		\draw[-, line width=5pt, opacity=0.2, color=cyan] ($(l)$)
		to ($(r)$) to ($(b)$) to cycle;
		\node at (-1.1*\xr,-1.1*\yr)[color=cyan]{$H$};
  }

  \newcommand{\PHTinhab}{
		\draw[-, line width=5pt, opacity=0.2, color=green!60!black] ($(l)+(4*\teps,-\teps)$)
		to ($(r)+(-4*\teps,-\teps)$) to ($(c)+(0,\teps)$) to cycle;
		\node at ($(c)+(0,0.4)$)[color=green!60!black]{$H_1$};
		\draw[-, line width=5pt, opacity=0.2, color=red] ($(r)+(-2.75*\teps,-2.5*\teps)$)
		to ($(b)+(0.5*\teps,2*\teps)$) to ($(c)+(0.5*\teps,-0.5*\teps)$) to cycle;
		\node at ($(c)+(0.25,-0.14)$)[color=red]{$H_2$};
  }

  \begin{scope}
   \labelk{a}{14}
   \PHTnodesF{}
   \PHTannot{0}
   \tikzESd{l/r/5,r/b/5,b/l/2}
   \PHTfixhab{}
   \PHTinhab{}
  \end{scope}

  \begin{scope}[xshift=1*\xsh*\xr cm]
   \labelk{b}{-1}
   \PHTnodes{}
   \PHTannot{1}
   \tikzESd{l/r/1,r/b/-3,b/l/2}
   \PHTfixhab{}
  \end{scope}

  \begin{scope}[xshift=2*\xsh*\xr cm]
   \labelk{c}{2}
   \PHTnodes{}
   \PHTannot{1}
   \draw[line width=2pt,color=red] (b) to (r);
   \tikzESd{l/r/1,r/b/0,b/l/2}
	 \PHTfixhab{}
  \end{scope}

 \end{tikzpicture}
 \caption{An illustration of the application of \cref{rr:reducibleH},
 with edge costs written on the edges and~$k$ given in the subfigure's label.
 (a) The habitats are~$H'=\{v_\ell,v_r,v_t\}$, $H=\{v_\ell,v_r,v_b\}$,
 $H_1=\{v_\ell,v_r,v_c\}$, and~$H_2=\{v_r,v_b,v_c\}$,
 where~$H$ is the reducible habitat with~$H_1$ and~$H_2$ being inside~$H$.
 (b) After deleting~$H_1$ and~$H_2$ and adjusting edge costs and~$k$.
 (c) After dealing with negative edge costs, where edge~$\{v_r,v_b\}$ (thick, red) is now forced.
 }
 \label{fig:planarhab3}
\end{figure}
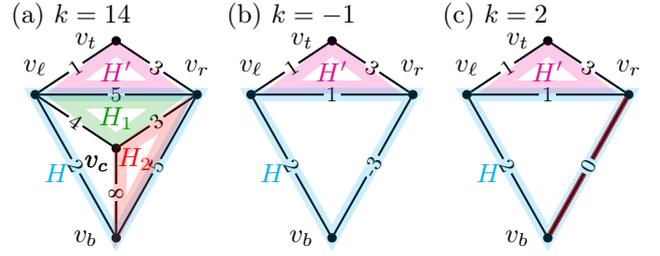

\begin{rrule}
	\label{rr:reducibleH}
	Let $H$ be a reducible habitat.
	Then delete $\ein{H}$ from the graph, $\hin{H}$ from the habitat set,
	increase $k$ by $2 \opt_H -\allowbreak \sum_{e \in E_H \setminus \setFc} \opt^e_H$,
	and increase the cost of each $e \in E_H$ by $\opt_H$.
	\emph{Afterwards}, decrease the cost of each $e \in E_H\setminus \setFc$ by $\opt_H^e$.
	For every edge $e$ whose cost $c(e)$ becomes negative in this process, increase $k$ by $|c(e)|$, set $c(e) \coloneqq 0$, and mark $e$ as forced.
\end{rrule}

\newcommand{\Fneg}{\ensuremath{F_{\mathrm{neg}}}}

\begin{proof}
	Denote by $\III$ and $\III'$ the instances before and after application of \cref{rr:reducibleH} and let $c'$ be the new cost function.
	Let us first determine the new bound $k'$ on the solution cost in $\III'$.
	Let $\Fneg \subseteq E_H \setminus \setFc$ be the set of edges whose cost was negative before the final step of the reduction rule.
	For every $e \in \Fneg$, we have $c(e) < \opt^e_H - \opt_H$.
	Then the set of forced edges in $\III'$ is $\setFc \cup \Fneg$, and $k' = k + \delta$, where
	\begin{equation}
		\label{eq:kprime}
		\delta = 2\opt_H - \!\!\!\sum_{e \in E_H \setminus \setFc} \opt^e_H + \sum_{e \in \Fneg} (\opt^e_H - \opt_H - c(e)).
	\end{equation}

	Let us now prove that $\III$ has a solution of cost at most $k$ if and only if $\III'$ has a solution of cost at most $k'$.
	Suppose first that $F$ is a solution of cost $c(F) \le k$ for $\III$.
	We will use the fact that
	\begin{equation}
		\label{eq:reducibleh-cost}
		c(F \cap \ein{H}) \ge
		\begin{cases}
			\opt_H, \text{ if } |F \cap E_H| = 3,\\
			\opt^e_H, \text{ if } E_H\setminus F = \{e\}.
		\end{cases}
	\end{equation}
	To verify this, observe that the inside optimum solution of cost $\opt_H$ is based on the inner cost function, which assigns cost $0$ to every edge in $E_H$,
	and the $e$-omitting inside optimum solution of cost $\opt^e_H$ is based on the $e$-omitting inner cost function, which assigns a prohibitively high cost to $e$.

	We claim that $F' \coloneqq (F \cup \Fneg) \setminus \ein{H}$ is a solution for $\III'$.
	Note that $c(e) \ne c'(e)$ only for all $e \in E_H$ (and of course, for $e \in \ein{H}$, $c'(e)$ is not defined).
	Thus,
	\begin{align*}
			c'(F') &= c(F) - c(F \cap \ein{H}) + \quad\quad \sum_{\mathclap{e \in F \cap (E_H \setminus (\setFc \cup \Fneg))}}\;\; (\opt_H - \opt^e_H)
			\\
			&\qquad + \sum_{e \in E_H \cap \setFc} \opt_H - \sum_{e \in F \cap \Fneg} c(e).\\
			      &\le k - c(F \cap \ein{H}) + |F \cap E_H| \opt_H - \!\!\!\!\!\!\!\!\!\!\sum_{e \in F \cap (E_H \setminus \setFc)}\!\!\!\!\!\!\!\!\!\! \opt^e_H
			      \\
			      &\qquad - \sum_{e \in F \cap \Fneg} (c(e) + \opt_H - \opt^e_H).\stepcounter{equation}\tag{\theequation}\label{eq:cost-of-f}
	\end{align*}

	If $|F \cap E_H| = 3$,
	then, using \eqref{eq:reducibleh-cost}, we have $c(F \cap \ein{H}) \ge \opt_H$ and $F \cap (E_H \setminus \setFc) = E_H \setminus \setFc$;
	by reformulating \eqref{eq:cost-of-f} we then obtain $c'(F') \le k + \delta = k'$.

	Otherwise, $E_H \setminus F = \{f\}$, and $c(F \cap \ein{H}) \ge \opt^f_H$.
	As $(F \cap (E_H \setminus \setFc)) \cup \{f\} = E_H \setminus \setFc$, we can again reformulate \eqref{eq:cost-of-f} to obtain $c'(F') \le k + \delta = k'$.

	Suppose next that $F'$ is a solution of cost $c'(F') \le k'$ for $\III'$.
	We will pick $F \coloneqq F' \cup F''$, where $F''$ will be later specified.
	Then, in analogy to the forward direction,
	\begin{align*}
		c(F) &= c(F'') + c'(F') - |F \cap E_H| \opt_H + \!\!\!\!\!\!\sum_{e \in F \cap (E_H \setminus \setFc)}\!\!\!\!\!\!\!\!\!\! \opt^e_H
		\\
		&\qquad + \!\!\sum_{e \in F \cap \Fneg} \!\!\!\!(c(e) + \opt_H - \opt^e_H).\stepcounter{equation}\tag{\theequation}\label{eq:cost-of-fprime}
	\end{align*}

	If $|F' \cap E_H| = 3$, then we pick $F''$ to be the inside optimum solution for $H$.
	Note that $F''$ fulfills the diameter constraint for all $H \in \hin{H}$ while $F'$ fulfills the diameter constraint for all remaining habitats $H$.
	Moreover, $c(F'') = \opt_H$ and $F \cap (E_H \setminus \setFc) = E_H \setminus \setFc$;
	thus reformulating \eqref{eq:cost-of-fprime} yields $c(F) = c'(F') - \delta \le k'-\delta = k$.

	Otherwise, $E_H \setminus F' = \{f\}$.
	Then, we pick $F''$ to be the $f$-omitting inside optimum solution for $H$.
	Analogously to the other case, $F$ fulfills the diameter constraint for all $H \in \HHH$.
	Moreover, $c(F'') = \opt^f_H$ and $(F \cap (E_H \setminus \setFc)) \cup \{f\} = E_H \setminus \setFc$;
	thus reformulating \eqref{eq:cost-of-fprime} yields $c(F) = c'(F') - \delta \le k'-\delta = k$.
\end{proof}

When applying \cref{rr:reducibleH} to a reducible habitat,
its parent becomes reducible.
Thus, after exhaustive application,
every habitat induces a face,
and we can apply \xcitet{Herkenrath et al.}{fluschnik2022placing}'s algorithm;
this proves \cref{prop:planarH3}.

\begin{proof}[Proof of~\cref{prop:planarH3}]
	We first apply \allrrs{} exhaustively in $\O(r\cdot \eta^3 + n + m) = \O(r + n)$ time, where $\eta$ is the maximum habitat size.
	We apply~\cref{rr:reducibleH} exhaustively.
	Note that \cref{rr:reducibleH} is applicable at most~$r$ times
	and its instance can be constructed in linear time.
	To apply \cref{rr:reducibleH},
	we need to find a reducible habitat
	and compute four solutions as given in the description of \cref{rr:reducibleH}.
	For a habitat~$H$,
	we can compute $\hin{H}$ in~$\O(r)$ time.
	Hence,
	to check whether a habitat~$H$ is reducible,
	we need~$\O(r^2)$ time.
	When we found a reducible habitat~$H$,
	we know that for every~$H'\in \hin{H}$,
	$\hin{H'}=\emptyset$ and hence,
	due to~\cref{lem:ExReducibleH},
	we know that each habitat hence induces a face.
	This case is solvable in $\O(n^2\cdot r)$ time \cite{fluschnik2022placing},
	and we need to find four solutions.
	Hence,
	in total we need~$\O(r\cdot(n^2\cdot r + r^2))$ time
	until \cref{rr:reducibleH} is inapplicable.
	Then there is no \reducible{} habitat,
	and thus,
	again,
	we can solve the remaining instance in~$\O(n^2\cdot r)$ time.
	In all, the running time is~$\O(n^2\cdot r^2 + r^3)$.
\end{proof}

\section{Maximum Degree Three}
\label{sec:deg3}
\newcommand{\mkthr}{\ensuremath{M(K_3)}}
\newcommand{\strngl}{zone}%
We next give a linear-time algorithm for subcubic graphs.
\begin{theorem}
 \label{thm:maxdeg3}
 \TwoDiamC{} is solvable in $\O(n+r)$ time on graphs of maximum degree three.
\end{theorem}

Due to \cref{rr:nontriangle}, we only care about those edges of habitats contained in a triangle.
Indeed, the degree and diameter constraints allow us to partition the triangles into edge-disjoint groups of at most two triangles,
which we call \emph{\strngl{s}}, see \cref{fig:strngl}.
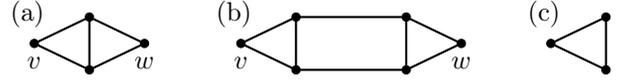
\begin{figure}[t!]
 \centering
 \begin{tikzpicture}
  \def\xr{0.725}
  \def\yr{0.35}
  \def\xsh{3.75}
  \tikzpreamble{}
  \newcommand{\tlabel}[1]{\node at (-1.125*\xr,1.125*\yr)[]{(#1)};}

  \begin{scope}
   \tlabel{a}
   \node (a) at (-1*\xr,0)[xnode,label=-90:{$v$}]{};
   \node (b) at (0*\xr,1*\yr)[xnode]{};
   \node (c) at (1*\xr,0)[xnode,label=-90:{$w$}]{};
   \node (d) at (0*\xr,-1*\yr)[xnode]{};
   \foreach \x/\y in {a/b,b/c,c/d,d/a,b/d}{\draw[xedge] (\x) to (\y);}
  \end{scope}
  \begin{scope}[xshift=1*\xsh*\xr cm]
   \tlabel{b}
   \node (a) at (-1*\xr,0)[xnode,label=-90:{$v$}]{};
   \node (b) at (0*\xr,1*\yr)[xnode]{};
   \node (c) at (0*\xr,-1*\yr)[xnode]{};
   \node (ap) at (-1*\xr+4*\xr,0)[xnode,label=-90:{$w$}]{};
   \node (bp) at (0*\xr+2*\xr,1*\yr)[xnode]{};
   \node (cp) at (0*\xr+2*\xr,-1*\yr)[xnode]{};
   \foreach \x/\y in {a/b,b/c,c/a,ap/bp,bp/cp,cp/ap,b/bp,c/cp}{\draw[xedge] (\x) to (\y);}
  \end{scope}
  \begin{scope}[xshift=2.5*\xsh*\xr cm]
   \tlabel{c}
   \node (a) at (-1*\xr,0)[xnode]{};
   \node (b) at (0*\xr,1*\yr)[xnode]{};
   \node (c) at (0*\xr,-1*\yr)[xnode]{};
   \foreach \x/\y in {a/b,b/c,c/a}{\draw[xedge] (\x) to (\y);}
  \end{scope}

 \end{tikzpicture}
 \caption{The \strngl{s}, where~(c)'s $K_3$ must not be contained in (a) or (b).}
 \label{fig:strngl}
\end{figure}
Every habitat intersects exactly one zone; moreover, any edge of a habitat outside of a zone is forced.
As the zones have constant size, we can brute-force an optimum solution for each,
the union of which then is an optimum solution for our instance; see \cref{alg:degthree}.
\newcommandx{\EHX}[2][1=H,2=X]{\ensuremath{E_{#1-#2}}}
\newcommandx{\FHX}[2][1=H,2=X]{\ensuremath{F_{#1,#2}}}
\newcommandx{\PHvw}[3][1=H,2=v,3=w]{\ensuremath{\calP_{#1,#2,#3}^{\leq 2}}}
\begin{algorithm}[t!]
	\caption{Given a reduced instance~$I=(G=(V,E),c,\calH,\setFc,k)$ where~$\Delta(G)\leq 3$,
	return a set~$\setFc\subseteq F\subseteq E$ or \texttt{False}.
Here, $\EHX\ceq E(G[H])\setminus E(X)$ and~$\FHX\ceq \EHX\cup F_X$.}\label[algorithm]{alg:degthree}
  Initialize~$F \gets \setFc$\;
	\ForEach{$X\subseteq G$ being a \strngl{}}{
		Find~$\setFc\cap E(X)\subseteq F_X\subseteq E(X)$ of lowest cost such that for each~$H\in \calH$ with~$H\cap V(X)\neq\emptyset$ it holds true that~$\diam(G[\FHX][H])\leq 2$
		$F\gets F\cup F_X$\;
	}
	\leIf{$c(F)\leq k$}{\textbf{return} $F$}{\textbf{return \texttt{False}}}
\end{algorithm}

Let~$\mkthr$ be the graph obtained from adding a complete matching between two~$K_3$'s.
Observe that, due to the degree bound, any $K_4$ and any $\mkthr$ would be isolated in~$G$ and thus removed by \cref{rr:smallcomps}.
Let~$K_4-e$ be a~$K_4$ where one edge is deleted.
Let $\mkthr-e_M$ be a~$\mkthr$ where one matching edge is deleted.

\begin{definition}
 \label{def:strngl}
A \emph{\strngl{}} is a subgraph of~$G$ isomorphic to a~$K_4-e$ or an~$\mkthr-e_M$, or a~$K_3$ that is not contained in a~$K_4-e$ or $\mkthr-e_M$; compare with \cref{fig:strngl}.
\end{definition}

Considering all six possible pairs of \strngl{s},
one can prove
that every two distinct zones are edge-disjoint.

\begin{observation}
 \label{obs:suptr:disjoint}
 Let~$X,Y$ be two distinct \strngl{s}.
 Then~$E(X)\cap E(Y) = \emptyset$.
\end{observation}

\newcommand{\iso}{\cong}

\begin{proof}
\xcase{1}{$X\iso K_3$ (i.e., $X$ is isomorphic to~$K_3$)}
If~$Y\iso K_3$,
then either they are not distinct or they form a~$K_4$,
a contradiction to both being a \strngl{}.
Let~$Y\iso K_4-e$ or $Y\iso \mkthr-e_M$.
Then~$Y$ has no edge with two endpoints each of degree smaller at most two.
Thus,
$X$ must be contained in~$Y$,
a contradiction.

\xcase{2}{$X\iso K_4-e$ and~$Y\not\iso K_3$}
Let $v, w \in X$ be the unique endpoints with~$\{v,w\}\notin E$.
Then $Y$ has a vertex not contained in~$V(X)$.
Since~$X$ and~$Y$ share an edge,
there must be~$x\in V(X)\setminus\{v,w\}$
contained in~$Y$.
Since~$K_4-e$ and~$\mkthr-e_M$ are 2-connected and contain no induced 5-cycle,
there must be a vertex~$u\in V(Y)\setminus V(X)$ and both edges~$\{v,u\},\{u,w\}\in E(Y)$.
Thus,
$u,v,w,x$ induce a four-cycle,
and hence,
$Y\not\iso K_4-e$.
Since also~$v,w,x$ has only~$y\in V(X)\setminus\{v,w,x\}$
as its only remaining neighbor,
$Y\not \iso \mkthr-e_M$.
Either way contradicts that~$Y\not\iso K_3$ is a \strngl.

\xcase{3}{$X,Y\iso \mkthr-e_M$}
Since~$G$ contains no $\mkthr$ we have~$\{v,w\}\not\in E$.
Thus, $Y$ has a vertex~not contained in~$V(X)$.
Since~$\mkthr-e_M$ is 2-connected,
there must be a vertex~$u\in V(Y)\setminus V(X)$ and a $(v,u)$-path and a $(w,u)$-path in~$Y$ which are disjoint.
Since~$X$ and~$Y$ share an edge,
there must be a cycle of length at least five in~$Y$,
a contradiction to $Y\iso \mkthr-e_M$.
\end{proof}

\noindent The following is the main property used by \cref{alg:degthree}.
\begin{observation}
 \label{obs:pathsincluded}
 Let~$H\in\calH$ and~$v,w\in H$ and~$F$ be a solution by~\cref{alg:degthree}.
 Let~$X\subseteq G$ be a \strngl{}
 such that~$P\subseteq X$ for all~$P\in\PHvw$,
 where~$\PHvw$ is the set of all $v$-$w$~paths of length at most two in~$G[H]$.
 Then~$\dist_{G[F][H]}(v,w)\leq 2$.
\end{observation}

\begin{proof}
 Let~$F_X = F\cap E(X)$.
 As~$\diam(G[F_{H,X}][H])\leq 2$,
 we have $\dist_{G[F_{H,X}][H]}(v,w)\leq 2$.
 All $v$-$w$~paths of length $\le 2$ are in~$X$;
 thus~$\dist_{G[F_{X}][H]}(v,w) = \dist_{G[F_{H,X}][H]}(v,w) \leq 2$.
\end{proof}

To see that \cref{alg:degthree} computes a solution, observe the following:
Either all relevant paths of length two between two habitat vertices are in a zone, in which case \cref{obs:pathsincluded} works,
or not, in which case all edges on a path outside the zone are forced.

\begin{proof}[Proof of \cref{thm:maxdeg3}]
 We first prove that~$F$ is a solution and then that it is of minimal cost.
 Finally, we show that \cref{alg:degthree} runs in $\O(n + r)$ time.
 Refer to \cref{fig:deg3:fsol} for an illustration.

 \emph{$F$ is a solution.}
 Let~$H\in \calH$ and let~$v,w\in H$.
 We show that $\dist_{G[F][H]}(v,w)\leq 2$.
 We perform an exhaustive case distinction on the $\PHvw$.
 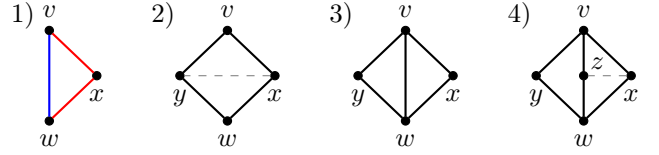
\begin{figure}[t!]
	\centering
	\begin{tikzpicture}
		\def\xr{1.25}
		\def\yr{0.6}
		\def\xsh{1.875}
		\tikzpreamble{}
		\newcommandx{\tlabel}[2][1=-0.675]{\node at (#1*\xr,1.325*\yr)[]{#2)};}

		\begin{scope}
		\tlabel[-0.275]{1}
		\node (a) at (0*\xr,1*\yr)[xnode,label=90:{$v$}]{};
		\node (b) at (0.5*\xr,0*\yr)[xnode,label=-90:{$x$}]{};
		\node (c) at (0*\xr,-1*\yr)[xnode,label=-90:{$w$}]{};
		\foreach \x/\y/\z in {a/b/red,b/c/red,a/c/blue}{\draw[xedge,color=\z] (\x) to (\y);}
		\end{scope}
		\begin{scope}[xshift=1*\xsh*\xr cm]
		\tlabel{2}
		\node (a) at (0*\xr,1*\yr)[xnode,label=90:{$v$}]{};
		\node (b) at (0.5*\xr,0*\yr)[xnode,label=-90:{$x$}]{};
		\node (bp) at (-0.5*\xr,0*\yr)[xnode,label=-90:{$y$}]{};
		\node (c) at (0*\xr,-1*\yr)[xnode,label=-90:{$w$}]{};
		\foreach \x/\y in {a/b,b/c,c/bp,bp/a}{\draw[xedge] (\x) to (\y);}
		\draw[dashed,gray] (b) to (bp);
		\end{scope}
		\begin{scope}[xshift=2*\xsh*\xr cm]
		\tlabel{3}
		\node (a) at (0*\xr,1*\yr)[xnode,label=90:{$v$}]{};
		\node (b) at (0.5*\xr,0*\yr)[xnode,label=-90:{$x$}]{};
		\node (bp) at (-0.5*\xr,0*\yr)[xnode,label=-90:{$y$}]{};
		\node (c) at (0*\xr,-1*\yr)[xnode,label=-90:{$w$}]{};
		\foreach \x/\y in {a/b,b/c,c/bp,bp/a,a/c}{\draw[xedge] (\x) to (\y);}
		\end{scope}

		\begin{scope}[xshift=3*\xsh*\xr cm]
		\tlabel{4}
		\node (a) at (0*\xr,1*\yr)[xnode,label=90:{$v$}]{};
		\node (b) at (0.5*\xr,0*\yr)[xnode,label=-90:{$x$}]{};
		\node (bp) at (-0.5*\xr,0*\yr)[xnode,label=-90:{$y$}]{};
		\node (bpp) at (0*\xr,0*\yr)[xnode,label={[label distance=-3pt]45:{$z$}}]{};
		\node (c) at (0*\xr,-1*\yr)[xnode,label=-90:{$w$}]{};
		\foreach \x/\y in {a/b,b/c,c/bp,bp/a,a/bpp,bpp/c}{\draw[xedge] (\x) to (\y);}
		\draw[dashed,gray] (b) to (bpp);
		\end{scope}

	\end{tikzpicture}
	\caption{Illustration to the cases in the proof of~\cref{thm:maxdeg3}.
	In~1), the cases are whether only the red, only the blue,
	or both paths are present.
	Dashed edges correspond to subcases.}
	\label{fig:deg3:fsol}
	\end{figure}

 \xcase{1}{$\PHvw$ consists of the paths~$(v,w)$ or~$(v,x,w)$}
 If $\PHvw$ consists of~$(v,w)$,
 then the edge is forced.
 If $\PHvw$ consists of both paths,
 then~$\PHvw\subseteq X$ for some \strngl{}~$X$
 and
 the claim follows from~\cref{obs:pathsincluded}.
 If $\PHvw$ consists of $P\ceq (v,x,w)$,
 then either each of the two edges is forced or one of them
 is in a triangle~$T\subseteq G[H]$.
 Then,
 $T\subseteq X$ for some \strngl{}~$X$
 and hence~$\diam(G[\FHX][H])\leq 2$.
 Since~$P$ is the only~$v$-$w$~path of length at most two,
 it follows that~$E(P)\subseteq F$.

 \xcase{2}{$\PHvw$ consists of the paths~$(v,x,w)$ and~$(v,y,w)$}
 If $\{x,y\}\in E$,
 then~$\PHvw\subseteq X$ for some \strngl{}~$X$
 and
 the claim follows from~\cref{obs:pathsincluded}.
 Let~$\{x,y\}\not\in E$.
 Assume towards a contradiction that $\dist_{G[F][H]}(u,v)>2$.
 Suppose that~$\{v,x\},\{v,y\}\not\in F$.
 Then both edges are contained a triangles.
 Since~$\Delta(G)\leq 3$,
 there is a vertex~$z$ adjacent to~$\{v,x,y\}$.
 Then~$\{v,x,y,z\}$ form a \strngl{},
 a contradiction to~$\diam(G[\FHX][H])\leq 2$ by~\cref{alg:degthree}.

 Suppose that~$\{v,x\},\{y,w\}\not\in F$.
 Then both edges are contained in a triangle,
 and since~$\Delta(G)\leq 3$,
 there are~$x',y'$ such that~$\{v,x,x'\}$ and~$\{w,y,y'\}$ form a triangle.
 Then,
 $\{v,x,x',w,y,y'\}$ form a \strngl{} (an~$\mkthr-e_m$),
 contradicting~\cref{obs:pathsincluded}.

 \xcase{3}{$\PHvw$ consists of the paths~$(v,w)$, $(v,x,w)$, and~$(v,y,w)$}
 Then~$\PHvw\subseteq X$ for some \strngl{}~$X$
 and
 the claim follows from~\cref{obs:pathsincluded}.

 \xcase{4}{$\PHvw$ consists of the paths~$(v,z,w), (v,x,w), (v,y,w)$}
 Note that since~$\Delta(G)\leq 3$,
 at most one of~$\{x,y\},\{y,z\},\{x,z\}$ is contained in~$E$.
 Let~$u\in\{x,y,z\}$ be without an edge to~$\{x,y,z\}\setminus\{u\}$.
 Then the path~$(v,u,w)$ contains only forced edges,
 and hence $\dist_{G[F][H]}(u,v)\leq 2$.

 \emph{$F$ is of minimum cost.}
 From above,
 we know that~$F$ is a solution.
 Assume towards a contradiction that there is a solution~$F'$ with~$c(F')<c(F)$.
 $F$ and~$F'$ only differ on edges contained in triangle.
 Since every triangle is contained in a \strngl{},
 there is an \strngl{}~$X$ with~$c(E(X)\cap F')<c(E(X)\cap F)$.
 This contradicts the choice of~$F_X$ by~\cref{alg:degthree}.

 Let us analyze the running time of the algorithm.
 Note that each habitat has constant size due to the degree constraint and the diameter constraint, or we have a trivial \no-instance.
 Thus, \allrrs{} are applicable in $\O(n + m + r)$ time.
 Since all \strngl{s} are edge-disjoint, there are most~$m$ many,
 which can be found in linear time by searching the constant-size 3-neighborhoods of each vertex~$v$.
 Since every habitat~$H$ is of constant size,
 every \strngl{}~$X$ is of constant size,
 and by the constant maximum degree,
 there are only constantly many habitats with a vertex in~$X$.
\end{proof}

\section{Maximum Degree Four}
\label{sec:deg4}

In this section, we will show the following result.

\begin{theorem}
	\label{cor:deg4}
	\TwoDiamC{} is solvable in $\O(n+r)$ time if the maximum degree of $G$ and the maximum habitat size are at most four.
\end{theorem}

Note that, even if the habitats have size at most four, we cannot remove habitats that are subsets of other habitats:
Consider a habitat $H = \{u,v,w,x\}$ inducing a $K_4$.
Then a solution $F$ with $\{u,v\}, \{u,w\}, \{u,x\} \in F$ (but no other edges in $G[H]$) is sufficient for $H$; however for the subset habitat $H' = \{v, w, x\}$, $F$ would be infeasible.
In our following algorithm, such habitats would be disturbing.
Thus, instead, we simply compute the set $\FFF_H$ of feasible solutions for each habitat $H \in \HHH$; note that each set $\FFF_H$ contains only constantly many solutions as $H$ is of constant size.
If we then have a habitat $H'$ which is a subset of some habitats $H \in \HHH'$, then we remove $H'$ and keep only those solutions in $\FFF_H$, $H \in \HHH'$ that are feasible for both $H$ and~$H'$.

We assume that \allrrs{} are inapplicable,
implying
that each habitat induces a $K_3$, a $K_4$, or a $\kfoure$.
Formally, we reduce \DwgbpAcr[2]{} to the following problem.

\decprob{\GwgbpTsc{} (\GwgbpAcr{})}{gwgbp}
{An undirected graph~$G=(V,E)$ with edge costs $c\colon E\to \Nzero$,
a set~$\calH=\{H_1,\dots,H_r\}$ of habitats with $H_i\subseteq V$ and~$|H_i|\geq 2$ for all~$i\in\set{r}$,
a set $\FFF_{H_i} \subseteq 2^{E(G[H_i])}$ of feasible edge sets for all $i \in \set{r}$,
a set~$\setFc\subseteq E$ of \emph{forced} edges,
and
an integer~$k\in\Nzero$.}
{Is there a subset~$F\subseteq E$ with
$\setFc\subseteq F$ and~
$c(F)\ceq \sum_{e\in F}c(e)\leq k$ such that
for every~$i\in\set{r}$
there is a set $F_{H_i} \in \FFF_{H_i}$ with $F_{H_i} \subseteq F$?
}

\begin{observation}
	\label{obs:diam-to-gen}
	For each instance $\I = (G, c, \HHH, \setFc, k)$ of \DwgbpAcr[2]{} with constant habitat size,
	one can compute in $\O(r)$ time the sets $\FFF_H \subseteq 2^{E(G[H])}$, $H \in \HHH$,
	so that $F$ is a solution for $\I$ if and only if $F$ is a solution for the instance $\I' = (G, c, \HHH, \{\FFF_H\}_{H \in \HHH}, \setFc, k)$ of \GwgbpAcr{}.
\end{observation}

We apply the following rule to \GwgbpAcr{} instances.

\begin{rrule}
	\label{rr:hab-subsets}
	Let $H' \in \HHH$ and let $\HHH' \ceq \{H \in \HHH \mid H' \subseteq H\}$.
	Then, for each $H \in \HHH'$,
	keep only those sets $F_H$ in $\FFF_H$ such that there is a set $F_{H'} \in \FFF_{H'}$ with $F_{H'} \subseteq F_H$,
	and delete $H'$ and $\FFF_{H'}$.
\end{rrule}

We next introduce the habitat intersection graph and show its structural properties.

\subsection{The habitat intersection graph}

The crucial part for computing a solution for \GwgbpAcr{} (and thus also \TwoDiam{})
is that a habitat can influence the subsolution for another habitat (directly) only if they both share an (unforced) edge.
However, if the two habitats have no common edges, but both have edges with a third habitat in common, then their subsolutions can still affect each other.

To exploit this property, we formally define the habitat intersection graph, where call two habitats
neighboring---up to one small caveat---when they share an unforced edge.

\begin{figure}
	\centering
	\tikzset{
	}
	\tikzpreamble{}
	\begin{tikzpicture}
		\def\centerdist{0.21}
		\def\xsh{10.5}
		\def\xxsh{1.1}
		\def\ysh{0.875}
		\def\ysh{0.875}
		\def\teps{0.2}
		\newcommand{\labelk}[1]{\node at (-0.9,0.5)[anchor=west]{(#1)};}
		\newcommand{\labelkk}[1]{\node at (-0.9,0.5)[anchor=west]{(#1)};}

		\begin{scope}
			\coordinate (t1) at (0,0) {};
			\path (t1) to ++(090:2*\centerdist) node [xnode] (u) {};
			\path (t1) to ++(210:2*\centerdist) node [xnode] (v) {};
			\path (t1) to ++(330:2*\centerdist) node [xnode] (w) {};
			\path (w)  to ++(090:2*\centerdist) coordinate (t2) {};
			\path (t2) to ++(30:2*\centerdist) node [xnode] (x) {};
			\path (x)  to ++(270:2*\centerdist) coordinate (t3) {};
			\path (t3) to ++(330:2*\centerdist) node [xnode] (y) {};
			\path (y)  to ++(090:2*\centerdist) coordinate (t3) {};
			\path (t3) to ++(30:2*\centerdist) node [xnode] (z) {};
			\draw (u) edge (v);
			\draw (v) edge (w);
			\draw (w) edge (u);
			\draw (u) edge (x);
			\draw[ultra thick] (w) edge (x);
			\draw (w) edge (y);
			\draw (x) edge (y);
			\draw (x) edge (z);
			\draw (y) edge (z);

			\draw[-, densely dashed, blue, rounded corners, draw] ($(v)+(-\teps,-\teps)$) to ($(w)+(\teps,-\teps)$) to ($(x)+(\teps,\teps)$)  to ($(u)+(0,\teps)$) --cycle; %
			\draw[-, thick, dotted, red, rounded corners] ($(x)+(-\teps,\teps)$) to ($(z)+(\teps,\teps)$) to ($(y)+(\teps,-\teps)$) to ($(w)+(-\teps,-\teps)$) --cycle; %
			\draw[-, thick, dash dot, violet, rounded corners] ($(u)+(-1.5*\teps,1.5*\teps)$) to ($(x)+(1.5*\teps,1.5*\teps)$) to ($(y)+(1.5*\teps,-1.5*\teps)$) to ($(w)+(-1.5*\teps,-1.5*\teps)$) --cycle; %
		\end{scope}

		\begin{scope}[xshift=3cm]
			\node [xnode] at (0,0) (h1) {};
			\node [xnode] at (1,0) (h2) {};
			\node [xnode] at (2,0) (h3) {};

			\draw (h1) edge (h2) edge (h3);

			\node [circle, densely dashed, blue, draw, inner sep=3pt, label=below:{\footnotesize $H_1$}] at (h1) {};
			\node [circle, dotted, thick, red, draw, inner sep=3pt, label=below:{\footnotesize $H_2$}] at (h2) {};
			\node [circle, dash dot, violet, thick, draw, inner sep=3pt, label=below:{\footnotesize $H_3$}] at (h3) {};
		\end{scope}

	\end{tikzpicture}
	\caption{
		Illustration of a graph with three habitats, and the corresponding intersection graph $\partial \HHH$ by \cref{def:intersection-graph}.
		Note that $H_1$ and $H_3$ are not neighboring in $\partial\HHH$ as $E(G[H_1 \cap H_3]) \subsetneq E(G[H_1 \cap H_2])$.
	}
	\label{fig:intersection-graph}
\end{figure}
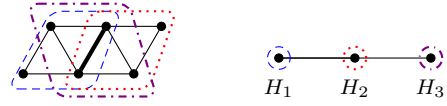
\begin{definition}
	\label{def:intersection-graph}
	For a graph $G$ with a subset of edges being forced and habitat set $\HHH$, the \emph{intersection graph}~$\partial \HHH$ has the vertex set $V(\partial\HHH) = \HHH$ and two habitats are \emph{neighboring}, i.e., $\{H_1, H_2\} \in E(\partial \HHH)$ if and only if
	\begin{compactenum}[(i)]
		\item $G[H_1 \cap H_2]$ contains at least one unforced edge and\label{def:intersection-graph:unforced}
		\item there is no $H_3 \in \HHH$ such that
			$E(G[H_1 \cap H_2])$ is a proper subset of $E(G[H_1 \cap H_3])$ \emph{or} $E(G[H_2 \cap H_3])$.\label{def:intersection-graph:proper}
	\end{compactenum}
	We call the edges in $G[H_1 \cap H_2]$ \emph{intersection edges} between~$H_1$ and~$H_2$.
	For a subset $\CCC \subseteq \HHH$,
	we denote by $V(\CCC)$ and $E(\CCC)$ the vertices and the edges of $G[\bigcup_{H \in \CCC} H]$.
\end{definition}

We show that
we can compute $\partial\HHH$ efficiently and solutions separately for each of its connected components.

\begin{observation}
	\label{obs:r-is-bounded}
	If the maximum habitat size and the maximum degree of the graph are at most four,
	then each edge is contained in at most 21 habitats.
\end{observation}

\begin{proof}
	Consider an edge $e = \{u, v\}$
	and two habitats $H$ and $H'$ containing $e$.
	Let $N(e) \ceq N(u)\cup N(v)$.
	As each habitat containing $e$ is a subset of $N(e)$,
	we have $H = H'$ whenever $N(e) \cap H = N(e) \cap H'$;
	i.e., each habitat $H$ containing $e$ is uniquely identified by $N(e) \cap H$.
	Since $|N(e)\setminus \{u, v\}| \le 6$,
	the number of habitats of size three and four containing~$e$ is at most~$6$ and~$\binom{6}{2}$,
	respectively.
\end{proof}

\begin{lemma}
 \label{lem:habintgr-rt}
 Given a graph $G$ with maximum degree four, a subset of forced edges and a set $\HHH$ of~$r$ habitats each of size at most four,
 one can compute~$\partial\HHH$ in~$\O(n+r)$ time.
\end{lemma}

\begin{proof}
	We first compute for each edge $e$ the set of habitats containing $e$ (which is constant by \cref{obs:r-is-bounded}) in two steps:
	First, compute for each vertex $v$ the set of habitats containing $v$ (which is also constant by \cref{obs:r-is-bounded}),
	then do a breadth-first search over $G$ and, whenever there is an edge $e = \{u,v\}$ with $u,v \in H$, add $H$ to the habitat set for $e$.
	Next, for every unforced edge $e$ and every pair $H_1, H_2$ of habitats containing $e$,
	we check \cref{def:intersection-graph}\eqref{def:intersection-graph:proper} by listing all habitats sharing an edge with $H_1$ or $H_2$ in constant time.
	If there is no $H_3 \in \HHH$ with $E(G[H_1 \cap H_2]) \subseteq E(G[H_i \cap H_3])$ for $i \in \{1, 2\}$, then add $\{H_1, H_2\}$ to $E(\partial\HHH)$.
	As there are $\O(n)$ edges, the claimed running time follows.
\end{proof}

\begin{lemma}
	\label{lem:partition-components}
	For an instance of \GwgbpAcr{},
	let $\CCC_1, \CCC_2, \dots, \CCC_\ell$ be the set of connected components of $\partial \HHH$.
	For each $i \in \{1, \dots, \ell\}$,
	define $G_i \coloneqq G[\bigcup \CCC_i]$,
	and $F_i \subseteq E(G_i)$ be a minimum-cost solution for the subgraph $G_i$.
	Then $F \coloneqq \setFc \cup F_1 \cup \dots \cup F_\ell$ is a minimum-cost solution for $G$, where $\setFc$ is the set of forced edges.
\end{lemma}

\begin{proof}
	If $G[H_1 \cap H_2]$ contains edges but $\{H_1, H_2\} \notin E(\partial\HHH)$,
	then all intersection edges are in $\setFc$,
	or there is a habitat $H_3$ as in \cref{def:intersection-graph}(ii).
	In the latter case, $\{H_1, H_3\}, \{H_3, H_2\} \in E(\partial\HHH)$; thus $H_1, H_2, H_3$ are in the same connected component.
	Thus, for every $i, j \in \{1, \dots, \ell\}$, we have $E(G_i) \cap E(G_j) \subseteq \setFc$.
	Consequently, $F_i \cap F_j \subseteq \setFc$.
	Clearly, $F$ is a solution, and every solution must contain $\setFc$.
	As every $F_i$ is cost-optimal for $G_i$, the resulting solution $F$ is also cost-optimal.
\end{proof}

\subsection{Structure of the habitat intersection graph}

The following important vertices limit how habitats intersect.

\begin{definition}
	\label{def:docking}
	For two neighboring habitats $H, H' \in \HHH$, we call a vertex $v \in H \cap H'$ \emph{docking towards $H'$} if it is adjacent to a vertex in $H' \setminus H$.
\end{definition}

Observe that a vertex $v \in H \cap H'$ may be docking towards $H'$, but not towards $H$; i.e., docking vertices are not symmetric.
Also, there are always at least two docking vertices towards $H'$ in $H \cap H'$ if $H$ and $H'$ are neighboring.

\begin{observation}
	\label{obs:docking-2-3}
	Let $H, H' \in \HHH$ be two neighboring habitats.
	Then there are at least two and at most $\max \{|H|, |H'|\}-1$ vertices in $H$ docking towards $H'$.
\end{observation}

\begin{proof}
	If there is only one docking vertex, then either $H$ and $H'$ do not have an intersecting edge and thus are not neighboring, or $H'$ is not $2$-connected.
	If there is none or the upper bound is violated, then one habitat is subset of the other.
\end{proof}

Note that the docking vertices towards $H'$ may have at most three neighbors within $H$, otherwise they cannot be adjacent to any vertex in $H' \setminus H$.
Indeed, we show that among the intersecting edges, there must be one whose endpoints both have degree at most three in $H$.
We call them \emph{docking edges}, and additionally \emph{proper}
if one of their endpoints has degree two.
We define these notions for sets of habitats.

\begin{definition}
	For a set $\CCC \subseteq \HHH$ of habitats,
	an edge $\{u,v\} \in E(\CCC)$ is called \emph{docking} if $\deg_{\CCC}(u) \le 3$ and $\deg_{\CCC}(v) \le 3$, and properly docking if $\deg_{\CCC}(u) = 2$ and $\deg_{\CCC}(v) = 3$.
\end{definition}

\begin{observation}
	\label{obs:potential-0}
	If a set $\CCC \subseteq \HHH$ does not contain docking edges, then $\CCC$ does not have neighboring habitats.
\end{observation}

\begin{proof}
	Suppose that there is a habitat $H \notin \CCC$ that is neighboring a habitat $H' \in \CCC$.
	Then there are two vertices $u, v \in H \cap H'$ that are docking towards $H$, and there is one vertex $x \in H \setminus V(\CCC)$.
	Thus, the degree of both of these vertices within $\CCC$ is at most three.
	So if $\{u,v\} \in E$, then $\CCC$ contains a docking edge.
	Hence, $\{u,v\} \notin E$, and as $H \cap H'$ contains an unforced edge, there is another vertex $w \in H \cap H'$, and $\{u,w\}, \{w, v\} \in E$.
	As $\CCC$ contains no docking edge, we have $\deg_\CCC(w) = 4$, and so $\{w, x\} \notin E$.
	Thus, $H$ induces a $C_4$ (or $G[H]$ contains even fewer edges);
	a contradiction to the inapplicability of \cref{rr:nontriangle}.
\end{proof}

We now show a monotonicity property when there are no proper docking edges in $\CCC$.

\begin{lemma}
	\label{lem:potential-deg3}
	Let $\CCC \subseteq \HHH$ be a set of at least two habitats that are connected in $\partial\HHH$.
	Let $H \in \HHH \setminus \CCC$ be adjacent to at least one habitat in $\CCC$.
	If for each vertex $v \in V(\CCC)$ docking towards $H$ we have $\deg_\CCC(v) \ge 3$,
	then $\CCC' \coloneqq \CCC \cup \{H\}$ contains strictly less docking edges than~$\CCC$.
	Moreover, $\CCC'$ does not contain a docking edge with an endpoint in $H$.
\end{lemma}

\begin{proof}
	Suppose first that $H \cap V(\CCC)$ contains exactly two vertices $u, v$.
	Then both
	are docking vertices, and $\deg_{\CCC}(u) = \deg_{\CCC}(v) = 3$.
	Moreover, $\{u,v\} \in E$, or $H$ is not neighboring any habitat in $\CCC$.
	If $H$ contains four vertices, then the two (adjacent) vertices in $H \setminus V(\CCC)$ both have degree two; thus $H$ induces a $C_4$; a contradiction to \cref{rr:nontriangle} being inapplicable.
	If $H = \{u, v, w\}$, then $\deg_{\CCC'}(w) = 3$ and $\deg_{\CCC'}(u) = \deg_{\CCC'}(v) = 4$.
	Thus, the new vertex $w \in V(\CCC') \setminus V(\CCC)$ is not incident to a docking edge,
	and $\{u,v\}$ is docking in $\CCC$ but not in $\CCC'$.
	As the remaining edges are unaffected,
	there are strictly less docking edges in $\CCC'$.

	Now, suppose that $H \cap V(\CCC)$ contains exactly three vertices $u,v,w$.
	Then $H$ contains a fourth vertex $x \notin V(\CCC)$.
	If one of the vertices in $H \cap V(\CCC)$, say $v$, is not docking, then it must be adjacent to both $u$ and $w$; otherwise $G[H]$ is not $2$-connected.
	Then $\{v,x\} \notin E$ (otherwise $v$ is docking); thus $u$, $v$, $w$ form a triangle in $G$ (otherwise $H$ induces a $C_4$).
	Then none of the edges incident to $x \in V(\CCC') \setminus V(\CCC)$ are docking,
	but $\{u,w\}$ is docking in $\CCC'$ but not in $\CCC$.
	As the remaining edges are unaffected, the number of docking edges in $\CCC'$ strictly decreases.
	Suppose next that $v$ is docking.
	If $u$, $v$, $w$ form a triangle in $G$ (i.e., $H$ induces a $K_4$), then we are in a case identical to the three-vertex case above.
	If one of the edges between $u, v, w$, say, $\{u,w\}$, is missing, then $H$ induces a $\kfoure$.
	In this case, $\{u,v\}$ and $\{v,w\}$ are docking in $\CCC$, but not in $\CCC'$; and there is no docking edge incident to $x$ as $\deg_{\CCC'}(u) = \deg_{\CCC'}(w) = 4$ and $\deg_{\CCC'}(x) = 3$.
	Thus, the number of docking edges strictly decreases.
\end{proof}

Next we show that, if $\CCC$ contains proper docking edges, then $\CCC$ is part of a path in $\partial\HHH$.

\begin{lemma}
	\label{lem:two-makes-a-path}
	Let $\CCC \subseteq \HHH$ be a set of at least two habitats that are connected in $\partial\HHH$.
	If $\CCC$ contains at least one proper docking edge, then $\partial\HHH[\CCC]$ is a path,
	and each habitat corresponding to an endpoint of the path contains exactly one degree-two vertex with a degree-three neighbor.
\end{lemma}

\begin{figure}
	\centering
	\tikzset{
	}
	\tikzpreamble{}
	\begin{tikzpicture}
		\def\centerdist{0.21}
		\def\xsh{10.5}
		\def\xxsh{1.1}
		\def\ysh{0.875}
		\def\ysh{0.875}
		\def\teps{0.2}
		\newcommand{\labelk}[1]{\node at (-0.9,0.5)[anchor=west]{(#1)};}
		\newcommand{\labelkk}[1]{\node at (-0.9,0.5)[anchor=west]{(#1)};}

		\begin{scope}
		  \labelk{a\textsubscript{1}}
			\coordinate (t1) at (0,0) {};
			\path (t1) to ++(090:2*\centerdist) node [xnode] (u) {};
			\path (t1) to ++(210:2*\centerdist) node [xnode] (v) {};
			\path (t1) to ++(330:2*\centerdist) node [xnode] (w) {};
			\path (w)  to ++(090:2*\centerdist) coordinate (t2) {};
			\path (t2) to ++(30:2*\centerdist) node [xnode] (x) {};
			\draw (u) edge (v);
			\draw (v) edge (w);
			\draw[ultra thick] (w) edge (u);
			\draw (u) edge (x);
			\draw (w) edge (x);

			\draw[-, densely dashed, blue, rounded corners, draw] ($(v)+(-\teps,-\teps)$)
		to ($(w)+(\teps,-\teps)$) to ($(u)+(0,\teps)$) to  ($(v)+(-\teps,-\teps)$);
		\draw[-, thick, dotted, red, rounded corners] ($(x)+(\teps,\teps)$)
		to ($(w)+(0,-\teps)$) to ($(u)+(-\teps,\teps)$) to ($(x)+(\teps,\teps)$) ;
		\end{scope}

		\begin{scope}[xshift=\xsh cm * \centerdist]
			\labelk{a\textsubscript{2}}
			\coordinate (t1) at (0,0) {};
			\path (t1) to ++(090:2*\centerdist) node [xnode] (u) {};
			\path (t1) to ++(210:2*\centerdist) node [xnode] (v) {};
			\path (t1) to ++(330:2*\centerdist) node [xnode] (w) {};
			\path (w)  to ++(090:2*\centerdist) coordinate (t2) {};
			\path (t2) to ++(30:2*\centerdist) node [xnode] (x) {};
			\path (x)  to ++(270:2*\centerdist) coordinate (t3) {};
			\path (t3) to ++(330:2*\centerdist) node [xnode] (y) {};
			\draw (u) edge (v);
			\draw (v) edge (w);
			\draw[ultra thick] (w) edge (u);
			\draw (u) edge (x);
			\draw (w) edge (x);
			\draw (w) edge (y);
			\draw (x) edge (y);

			\draw[-, densely dashed, blue, rounded corners, draw] ($(v)+(-\teps,-\teps)$)
		to ($(w)+(\teps,-\teps)$) to ($(u)+(0,\teps)$) to  ($(v)+(-\teps,-\teps)$);
		\draw[-, thick, dotted, red, rounded corners] ($(x)+(\teps,\teps)$) to ($(y)+(\teps,-\teps)$) to ($(w)+(0,-\teps)$) to ($(u)+(-\teps,\teps)$) to ($(x)+(\teps,\teps)$) ;
		\end{scope}

		\begin{scope}[xshift=2*\xsh cm * \centerdist]
		  \labelk{a\textsubscript{3}}
			\coordinate (t1) at (0,0) {};
			\path (t1) to ++(090:2*\centerdist) node [xnode] (u) {};
			\path (t1) to ++(210:2*\centerdist) node [xnode] (v) {};
			\path (t1) to ++(330:2*\centerdist) node [xnode] (w) {};
			\path (w)  to ++(090:2*\centerdist) coordinate (t2) {};
			\path (t2) to ++(30:2*\centerdist) node [xnode] (x) {};
			\path (x)  to ++(270:2*\centerdist) coordinate (t3) {};
			\path (t3) to ++(330:2*\centerdist) node [xnode] (y) {};
			\draw (u) edge (v);
			\draw (v) edge (w);
			\draw[ultra thick] (w) edge (u);
			\draw[ultra thick] (u) edge (x);
			\draw[ultra thick] (w) edge (x);
			\draw (w) edge (y);
			\draw (x) edge (y);

			\draw[-, densely dashed, blue, rounded corners, draw] ($(v)+(-\teps,-\teps)$)
		to ($(w)+(\teps,-\teps)$) to ($(x)+(\teps,\teps)$)  to ($(u)+(0,\teps)$) to  ($(v)+(-\teps,-\teps)$);
		\draw[-, thick, dotted, red, rounded corners] ($(x)+(\teps,\teps)$) to ($(y)+(\teps,-\teps)$) to ($(w)+(0,-\teps)$) to ($(u)+(-\teps,\teps)$) to ($(x)+(\teps,\teps)$);
		\end{scope}

		\begin{scope}[yshift=-7*\ysh cm*\centerdist,xshift=0*\xsh cm * \centerdist]
			\labelk{a\textsubscript{4}}
			\coordinate (t1) at (0,0) {};
			\path (t1) to ++(090:2*\centerdist) node [xnode] (u) {};
			\path (t1) to ++(210:2*\centerdist) node [xnode] (v) {};
			\path (t1) to ++(330:2*\centerdist) node [xnode] (w) {};
			\path (w)  to ++(090:2*\centerdist) coordinate (t2) {};
			\path (t2) to ++(30:2*\centerdist) node [xnode] (x) {};
			\path (x)  to ++(270:2*\centerdist) coordinate (t3) {};
			\path (t3) to ++(330:2*\centerdist) node [xnode] (y) {};
			\path (y)  to ++(090:2*\centerdist) coordinate (t3) {};
			\path (t3) to ++(30:2*\centerdist) node [xnode] (z) {};
			\draw (u) edge (v);
			\draw (v) edge (w);
			\draw (w) edge (u);
			\draw (u) edge (x);
			\draw[ultra thick] (w) edge (x);
			\draw (w) edge (y);
			\draw (x) edge (y);
			\draw (x) edge (z);
			\draw (y) edge (z);

			\draw[-, densely dashed, blue, rounded corners, draw] ($(v)+(-\teps,-\teps)$)
		to ($(w)+(\teps,-\teps)$) to ($(x)+(\teps,\teps)$)  to ($(u)+(0,\teps)$) to  ($(v)+(-\teps,-\teps)$);
		\draw[-, thick, dotted, red, rounded corners] ($(x)+(-\teps,\teps)$) to ($(z)+(\teps,\teps)$) to ($(y)+(\teps,-\teps)$) to ($(w)+(-\teps,-\teps)$)to ($(x)+(-\teps,\teps)$);
		\end{scope}

		\begin{scope}[yshift=-7*\ysh cm*\centerdist,xshift=16*\xxsh cm*\centerdist]
			\labelkk{b\textsubscript{1}}
			\coordinate (t1) at (0,0) {};
			\path (t1) to ++(090:2*\centerdist) node [xnode] (u) {};
			\path (t1) to ++(210:2*\centerdist) node [xnode] (v) {};
			\path (t1) to ++(330:2*\centerdist) node [xnode] (w) {};
			\path (w)  to ++(090:2*\centerdist) coordinate (t2) {};
			\path (t2) to ++(30:2*\centerdist) node [] (x) {};
			\path (x)  to ++(270:2*\centerdist) coordinate (t3) {};
			\path (t3) to ++(330:2*\centerdist) node [] (y) {};
			\path (y)  to ++(090:2*\centerdist) coordinate (t3) {};
			\path (t3) to ++(30:2*\centerdist) node [xnode] (z) {};
			\path (z)  to ++(270:2*\centerdist) coordinate (t4) {};
			\path (t4) to ++(330:2*\centerdist) node [xnode] (a) {};
			\path (a)  to ++(090:2*\centerdist) coordinate (t5) {};
			\path (t5) to ++(30:2*\centerdist) node [xnode] (b) {};
			\draw (u) edge (v);
			\draw (v) edge (w);
			\draw (w) edge (u);

			\foreach \s/\t in {w/x,u/x,w/y,w/y,z/x,z/y,a/y}{
				\path (\s.center) to node[pos=0.30] (1\s) {} node[pos=0.70] (2\s) {} (\t.center);
				\draw (\s) edge (1\s.center);
				\draw (1\s.center) edge[dotted] (2\s.center);
			}

			\draw (b) edge (z);
			\draw[ultra thick] (a) edge (z);
			\draw (a) edge (b);

			\draw[-, densely dashed, blue, rounded corners, draw] ($(a)+(0,-\teps)$)
		to ($(b)+(\teps,\teps)$) to ($(z)+(-\teps,\teps)$) to ($(a)+(0,-\teps)$);
		\end{scope}

		\begin{scope}[yshift=-2*7*\ysh cm*\centerdist,xshift=0*\xxsh cm*\centerdist]
			\labelkk{b\textsubscript{2}}
			\coordinate (t1) at (0,0) {};
			\path (t1) to ++(090:2*\centerdist) node [xnode] (u) {};
			\path (t1) to ++(210:2*\centerdist) node [xnode] (v) {};
			\path (t1) to ++(330:2*\centerdist) node [xnode] (w) {};
			\path (w)  to ++(090:2*\centerdist) coordinate (t2) {};
			\path (t2) to ++(30:2*\centerdist) node [] (x) {};
			\path (x)  to ++(270:2*\centerdist) coordinate (t3) {};
			\path (t3) to ++(330:2*\centerdist) node [] (y) {};
			\path (y)  to ++(090:2*\centerdist) coordinate (t3) {};
			\path (t3) to ++(30:2*\centerdist) node [xnode] (z) {};
			\path (z)  to ++(270:2*\centerdist) coordinate (t4) {};
			\path (t4) to ++(330:2*\centerdist) node [xnode] (a) {};
			\path (a)  to ++(090:2*\centerdist) coordinate (t5) {};
			\path (t5) to ++(30:2*\centerdist) node [xnode] (b) {};
			\path (b)  to ++(270:2*\centerdist) coordinate (t6) {};
			\path (t6) to ++(330:2*\centerdist) node [xnode] (c) {};
			\draw (u) edge (v);
			\draw (v) edge (w);
			\draw (w) edge (u);

			\foreach \s/\t in {w/x,u/x,w/y,w/y,z/x,z/y,a/y}{
				\path (\s.center) to node[pos=0.30] (1\s) {} node[pos=0.70] (2\s) {} (\t.center);
				\draw (\s) edge (1\s.center);
				\draw (1\s.center) edge[dotted] (2\s.center);
			}

			\draw[ultra thick] (b) edge (z);
			\draw[ultra thick] (a) edge (z);
			\draw[ultra thick] (a) edge (b);
			\draw (a) edge (c);
			\draw (b) edge (c);

			\draw[-, densely dashed, blue, rounded corners, draw] ($(a)+(0,-\teps)$)
		to ($(c)+(\teps,-\teps)$) to ($(b)+(\teps,\teps)$) to ($(z)+(-\teps,\teps)$) to ($(a)+(0,-\teps)$);
		\end{scope}

		\begin{scope}[yshift=-2*7*\ysh cm*\centerdist,xshift=16*\xxsh cm*\centerdist]
			\labelkk{b\textsubscript{3}}
			\coordinate (t1) at (0,0) {};
			\path (t1) to ++(090:2*\centerdist) node [xnode] (u) {};
			\path (t1) to ++(210:2*\centerdist) node [xnode] (v) {};
			\path (t1) to ++(330:2*\centerdist) node [xnode] (w) {};
			\path (w)  to ++(090:2*\centerdist) coordinate (t2) {};
			\path (t2) to ++(30:2*\centerdist) node [] (x) {};
			\path (x)  to ++(270:2*\centerdist) coordinate (t3) {};
			\path (t3) to ++(330:2*\centerdist) node [] (y) {};
			\path (y)  to ++(090:2*\centerdist) coordinate (t3) {};
			\path (t3) to ++(30:2*\centerdist) node [xnode] (z) {};
			\path (z)  to ++(270:2*\centerdist) coordinate (t4) {};
			\path (t4) to ++(330:2*\centerdist) node [xnode] (a) {};
			\path (a)  to ++(090:2*\centerdist) coordinate (t5) {};
			\path (t5) to ++(30:2*\centerdist) node [xnode] (b) {};
			\path (b)  to ++(270:2*\centerdist) coordinate (t6) {};
			\path (t6) to ++(330:2*\centerdist) node [xnode] (c) {};
			\draw (u) edge (v);
			\draw (v) edge (w);
			\draw (w) edge (u);

			\foreach \s/\t in {w/x,u/x,w/y,w/y,z/x,z/y,a/y}{
				\path (\s.center) to node[pos=0.30] (1\s) {} node[pos=0.70] (2\s) {} (\t.center);
				\draw (\s) edge (1\s.center);
				\draw (1\s.center) edge[dotted] (2\s.center);
			}

			\draw (b) edge (z);
			\draw[ultra thick] (a) edge (z);
			\draw (a) edge (b);
			\draw (a) edge (c);
			\draw (b) edge (c);

		\draw[-, densely dashed, blue, rounded corners, draw] ($(a)+(0,-\teps)$)
		to ($(c)+(\teps,-\teps)$) to ($(b)+(\teps,\teps)$) to ($(z)+(-\teps,\teps)$) to ($(a)+(0,-\teps)$);
		\end{scope}

	\end{tikzpicture}
	\caption{
		Illustration to \cref{lem:two-makes-a-path} with base cases (a\textsubscript{1})--(a\textsubscript{4}), with two habitats (dashed/dotted), and induction step cases (b\textsubscript{1})--(b\textsubscript{3}) with habitat~$H$ (dashed). Thick edges are shared.
	}
	\label{fig:two-makes-a-path}
\end{figure}
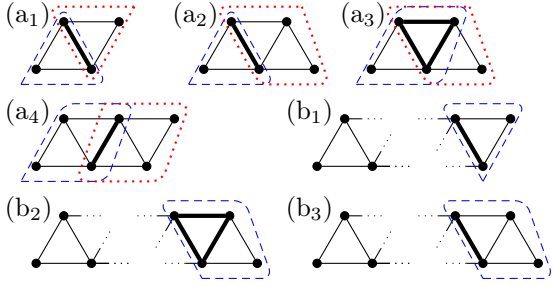

\begin{proof}
	We prove the statement by induction on the number $|\CCC|$ of habitats in $\CCC$.
	Our base case is $|\CCC| = 2$.
	Then, $\partial\HHH[\CCC]$ is a path.
	We now show that each habitat in $\CCC$ contains exactly one degree-two vertex which is endpoint to a proper docking edge.
	If one habitat in $\CCC$ induces a $K_4$, then $\CCC$ does not contain any proper docking edge by \cref{lem:potential-deg3}.
	So the two habitats can only induce a $K_3$ or a $\kfoure$.
	If both habitats induce a $K_3$, then they share exactly one edge (or the two habitats are equal).
	Thus, $\CCC$ induces a $\kfoure$ in $G$, which fulfills the asked properties.
	If one habitat induces a $K_3$ and the other induces a $\kfoure$, then again the two share exactly one edge.
	If the shared edge is not a proper docking edge of the $\kfoure$ habitat, then $\CCC$ has no proper docking edges.
	So the shared edge must be any other edge of the $\kfoure$ habitat.
	The resulting habitat set $\CCC$ contains exactly one proper docking edge per habitat, and the fifth vertex has degree four.
	If both habitats induce a $\kfoure$, then $5 \le |V(\CCC)| \le 6$.
	If $|V(\CCC)| = 5$, then the habitats intersect in a triangle, and $\CCC$ as the same structure as in the previous case.
	If $|V(\CCC)| = 6$, then the two habitats again share just the edge, and it must be a proper docking edge in both habitats, or \cref{lem:potential-deg3} applies. %
	Thus, $\CCC$ contains one proper docking edge per habitat. Its remaining two vertices have degree four.

	For the induction step, suppose that the statement is true for any set of at most $\ell$ habitats.
	Let $\CCC' \subseteq \HHH$ be a set of $\ell+1$ such that $\partial\HHH[\CCC']$ is connected and $\CCC'$ contains a habitat $H$ with a proper docking edge.
	Define $\CCC \coloneqq \CCC' \setminus \{H\}$.

	If $V(\CCC) = V(\CCC')$, then $\CCC$ is connected in $\partial\HHH$ and contains at least one proper docking edge; thus $\CCC'$ satisfies the induction hypothesis, and we are done.
	So suppose that $V(\CCC) \subset V(\CCC')$.
	As $\CCC'$ is connected, there is at least one habitat $H' \in \CCC$ that is neighboring $H$.
	As $3 \le |H| \le 4$ and as $H$ and $H'$ share at least one edge,
	we have $2 \le |H \cap H'| \le 3$, and that $1 \le |H \setminus V(\CCC)| \le 2$.
	Moreover, as $H$ contains a proper docking edge,
	\cref{lem:potential-deg3} implies that there is a vertex $u \in H' \cap H$ with $\deg_{\CCC}(u) = 2$ that is docking towards $H$.
	Moreover, there is a second docking vertex $v \in H' \cap H$ that is adjacent to $u$.
	Observe that $\deg_{\CCC}(v) = 3$ (if its degree were $2$, then $H'$ induces a $C_4$ or $|V(\CCC)| = 3$);
	thus $\{u,v\}$ is a proper docking edge in $\CCC$.
	By the induction hypothesis, this implies that $H'$ corresponds to an endpoint of the path induced by $\CCC$ in $\partial\HHH$.
	So, to show that $\CCC'$ induces a path in $\partial\HHH$ as well, we must show that there is no second habitat neighboring $H$.
	Assume towards a contradiction that there is a habitat  $H'' \in \CCC$ neighboring $H$.
	We do a case distinction on
	$|V(\CCC) \cap H|$.

	Suppose first that $V(\CCC) \cap H = \{u,v\}$.
	Then $H \cap H' = H \cap H'' = \{u,v\}$, and as $H \cap H''$ must not be a proper subset of $H' \cap H''$, we have $H' \cap H'' = \{u,v\}$ as well.
	But then, as both $H'$ and $H''$ are $2$-connected and contain at least one more vertex, $u$ and $v$ are together incident to four edges within $\CCC$ whose second endpoint is not $u$ or $v$.
	Thus, $\deg_{\CCC}(u) + \deg_{\CCC}(v) \ge 6$; a contradiction to $\{u,v\}$ being a proper docking edge in $\CCC$.

	Suppose next that $V(\CCC) \cap H = \{u,v,w\}$.
	Then $|V(\CCC)| \ge 4$ as otherwise $V(\CCC) \subseteq H$, and also $H$ contains a fourth vertex $x \notin V(\CCC)$.
	Moreover, $u$, $v$ and $w$ induce a triangle (otherwise, either $H$ induces a $C_4$ or $x$ is adjacent to all three vertices and $H$ does not contain a proper docking edge; note that this also implies that $w$ is not docking towards $H$).
	If $H \cap H' = \{u,v\}$,
	then $w \in H'' \setminus H'$.
	But then, as in the previous case, we have $\deg_{\CCC}(u) + \deg_{\CCC}(v) \ge 6$ and obtain a contradiction.
	Symmetry implies that $H \cap H'' \ne \{u,v\}$.
	So suppose that $H \cap H' = H \cap H'' = \{u,v,w\}$.
	Then there are vertices $y \in H' \setminus H$
	and $z \in H'' \setminus H$, and $y \ne z$.
	Moreover, $y$ and $z$ must be adjacent to at least two of $u$, $v$ and $w$.
	This again implies that $\deg_{\CCC}(u) + \deg_{\CCC}(v) \ge 6$, and we again obtain our contradiction.

	This proves that $\CCC'$ induces a path in $\partial\HHH$ as well.
	Moreover, the only vertices $z \in V(\CCC')$ with $\deg_{\CCC}(z) < \deg_{\CCC'}(z)$ are the endpoints of the proper docking edge $\{u,v\}$.
	By the induction hypothesis, $\CCC$ contains another proper docking edge $e$ whose endpoint degrees are unchanged in $\CCC'$,
	the habitat set $\CCC'$ also contains two proper docking edges.
	This finishes the proof.
\end{proof}

From this, we can deduce the following.

\begin{lemma}%
	\label{lem:path-cycle-constant}
	Each component in $\partial\HHH$ induces a path or a cycle, or has constant size.
\end{lemma}

\begin{proof}
	We show that if there is a habitat $H$ with at least three neighbors in $\partial\HHH$, then it is part of a constant-size component~$\CCC'$.
	Let $\CCC$ comprise of $H$ and its neighbors.
	By the contraposition of \cref{lem:two-makes-a-path}, $\CCC$ does not contain any proper docking edge.
	So any habitat added to $\CCC$ strictly decreases the number of docking edges of the subgraph by \cref{lem:potential-deg3}, and by \cref{obs:potential-0}, $\CCC$ cannot grow without docking edges.
	As $|E(\CCC)|$ initially is a constant, its number of docking edges is also constant; thus $|\CCC'|$ is also constant.
\end{proof}

\subsection{The algorithm}

We give an algorithm for \GwgbpAcr{} where $\partial\HHH$ is a path.

\newcommand{\mxfh}{\ensuremath{\lambda}}

\begin{lemma}
	\label{thm:dp-path}
	Let $\I$ be an instance of \GwgbpAcr{} where $\partial\HHH$ is a path and \cref{rr:hab-subsets} is exhaustively applied.
	If $\partial\HHH$ is given, then one can compute in $\O(r\cdot \mxfh^4)$ time an optimum cost solution for~$\I$.
	Here, $\mxfh \coloneqq \max_{H \in \HHH} |\FFF_H|$.
\end{lemma}

\begin{proof}
	Let $(H_1, H_2, \dots, H_r)$ be the path in $\partial\HHH$.
	Assume without loss of generality that for each $H \in \HHH$ and each $F_H \in \FFF_H$, we have $\setFc \cap E(G[H]) \subseteq F_H$, i.e., forced edges are parts of all subsolutions.

	Let $D[i, F_1, F_2, F_3]$ denote the costs when using subsolutions $F_1 \in \FFF_{H_i}, F_2 \in \FFF_{H_{i+1}}, F_3 \in \FFF_{H_{i+2}}$ for habitats $H_i, H_{i+1}, H_{i+2}$, respectively, under the assumption that the cost for the subsolutions for all habitats $H_1, \dots, H_{i+1}$ is minimal given $F_1$ and $F_2$.
	Given the values of all entries $D[j, \cdot, \cdot, \cdot]$ for $j < i$, we can compute the entry as follows.
\[
D[i,F_1,F_2,F_3] = \begin{cases}
	c(F_1\cup F_2\cup F_3), \text{if $i=1$,}& \\
	\min\limits_{F_0\in \FFF_{H_{i-1}}} D[i-1,F_0,F_1,F_2]
	& \\
	\qquad
	+ c(F_3\setminus (F_1\cup F_2)), \text{otherwise.}&
\end{cases}
\]
The optimum cost for the instance then is
\[
	\min_{F_1\in \FFF_{H_{r-2}},F_2\in \FFF_{H_{r-1}},F_3\in \FFF_{H_r}} D[r-2, F_1, F_2, F_3].
\]

The deciding observation to see that this dynamic programming approach is correct is that for $i \in [r-2]$, the habitats $H_{i-1}$ and $H_{i+2}$ do not have a common edge.
This is true since $H_{j+1} \setminus H_{j}$ and $H_{j} \setminus H_{j+1}$ each contain at least one vertex for each $j \in [r-1]$ and each habitat contains at most four vertices.
Thus, $F_{i-1} \cap F_{i+2}$ for any two $F_{i-1} \in \FFF_{H_{i-1}}$ and $F_{i+2} \in \FFF_{H_{i+2}}$.

The running time to compute an entry of $D$ is $\O(\mxfh)$ time.
As there are at most $\mxfh^3$ entries for each $i \in [r-2]$,
the overall running time is $\O(r \cdot \mxfh^4)$.
Using standard dynamic programming backtracking techniques, one can compute a solution with optimal cost.
\end{proof}

With this at hand, we can prove the section's theorem.
\begin{proof}[Proof of \cref{cor:deg4}]
	We first apply \allrrs{} in $\O(n + r)$ time.
	Then, we reduce our instance of \TwoDiamC{} to an equivalent instance of \GwgbpAcr{} (\cref{obs:diam-to-gen}).
	We then compute $\partial\HHH$ in~$\O(n + r)$ time (\cref{lem:habintgr-rt}).
	By \cref{lem:partition-components}, it suffices to compute a solution for each connected component $\CCC$ of $\partial\HHH$.
	By \cref{lem:path-cycle-constant}, each component induces a path or a cycle, or has constant size.
	If $\CCC$ has constant size, then we
	brute-force over all solutions in constant time.
	If $\CCC$ is a path, then we use \cref{thm:dp-path} to compute a solution; by \cref{thm:dp-path}, this takes $\O(|\CCC|\cdot \mxfh^4)$ time,
	where $\mxfh \coloneqq \max_{H \in \HHH} |\FFF_H|$.
	If $\CCC$ is a cycle, then we take any habitat $H \in \CCC$, guess a subsolution $F_H \in \FFF_H$ (note that there are constantly many solutions to pick from), mark all edges $e \in F_H$ as forced (i.e., add them to $\setFc$), and solve the remaining path component $\CCC \setminus \{H\}$ using \cref{thm:dp-path}.
	This takes $\O(\mxfh \cdot |\CCC| \cdot \mxfh^4)$ time.

	By \cref{lem:partition-components}, the union of the computed solutions is optimal for our instance of \GwgbpAcr{}, and thus also for our input instance.
	As $\mxfh$ is constant, our algorithm's running time is in $\O(n + r)$.
\end{proof}

\section{Maximum Degree Five}
\label{sec:maxdeg5}

We know that \TwoDiamC{} is polynomial-time solvable when
\begin{enumerate}[(a)]
\item the graph has maximum degree four and maximum habitat size three (\cref{cor:deg4}),
\item the graph is planar and the maximum habitat size is three (\cref{prop:planarH3}), or
\item the maximum habitat size is two (as every edge contained in a habitat is forced).
\end{enumerate}
We will prove these findings to be tight by proving \NP-hardness for \TwoDiam{},
whenever
in (a) the maximum degree or
in (b) and (c) the habitat size
is increased by one.
We prove \eqref{thm:nph:deg5:hab3} and \eqref{thm:nph:deg5:planarhab4}

\begin{theorem}
	\label{thm:nph:deg5}
	\TwoDiam{} is \nphard{} even if
	\begin{inparaenum}[(i)]
	\item the graph has maximum degree five and the maximum habitat size is three, or \label{thm:nph:deg5:hab3}
	\item the graph is planar, has maximum degree five, and the maximum habitat size is four.\label{thm:nph:deg5:planarhab4}
	\end{inparaenum}
\end{theorem}

We prove \eqref{thm:nph:deg5:hab3} and \eqref{thm:nph:deg5:planarhab4}
via polynomial-time reductions from the \nphard{} \cite{garey1976simplified} \PCVC{} problem, where,
given a planar, connected graph $G = (V, E)$ in which every vertex has degree three and an integer $k$,
the task is to decide whether $G$ has a vertex cover of size at most $k$.

In the following two constructions (\cref{constr:nph-deg-5,constr:nph-planar}), we use the following \emph{docking} operation on graphs.

\begin{definition}
	\label{def:nph-docking}
	Let $G = (V, E)$ and $G' = (V, E)$ be two graphs and let $d = (u,v)$, $\{u, v\} \in E$ and $d' = (u', v')$, $\{u', v'\} \in E'$.
	The procedure of \emph{docking} $G$ and $G'$ at $d$ and $d'$ is the same as identifying $u$ with $u'$ and $v$ with $v'$, i.e., it results in the graph  $G^* = (V^*, E^*)$ with
	$V^* \coloneqq V \cup V' \setminus \{u', v'\}$ and
	$E^* \coloneqq E \cup \{ \{u, y'\} \mid y' \in N_{G'}(u') \}
		\cup \{ \{v, y'\} \mid y' \in N_{G'}(v') \}
		\cup E(G' - \{u', v'\})$.
\end{definition}

Observe the following about the degrees of vertices after being glued together.

\begin{observation}
	\label{obs:nph-docking-deg}
	For two graphs $G = (V, E)$ and $G' = (V, E)$,
	let $G^*$ be the graph obtained from docking $G$ and $G'$ at $(u, v)$ and $(u', v')$.
	Then $\deg_{G^*}(u) = \deg_{G}(u) + \deg_{G'}(u') - 1$
	and $\deg_{G^*}(v) = \deg_{G}(v) + \deg_{G'}(v') - 1$.
\end{observation}

We first prove \cref{thm:nph:deg5}\eqref{thm:nph:deg5:hab3}.
For this, we employ the following construction.

	\begin{figure}[t]
		\centering
		\begin{tikzpicture}
		\def\xr{0.675}
		\def\yr{0.575}
		\def\xsh{3.525}
		\tikzpreamble{}

		\begin{scope}
			\node (v0) at (1*\xr,0*\yr)[xnode]{};
			\node (v1) at (0*\xr,1*\yr)[xnode]{};
			\node (v2) at (2*\xr,1*\yr)[xnode]{};
			\node (v3) at (2*\xr,-1*\yr)[xnode]{};
			\node (v4) at (0*\xr,-1*\yr)[xnode]{};

			\node at (v0) [label={[label distance=-2pt]270:{\small $b_u^5$}}] {};
			\node at (v1) [label={[label distance=-3pt]90:{$b_u^1$}}] {};
			\node at (v2) [label={[label distance=-3pt]090:{$b_u^2$}}] {};
			\node at (v3) [label={[label distance=-3pt]270:{$b_u^3$}}] {};
			\node at (v4) [label={[label distance=-3pt]270:{$b_u^4$}}] {};

			\tikzES{v0/v1,v0/v3}
			\tikzES[-,thick,dashed,decoration={markings,mark=at position 0.5 with {\arrow{>}}},postaction={decorate}]{v0/v4,v2/v1,v2/v3}
			\tikzES[-,ultra thick, color=red]{v1/v4,v4/v3,v0/v2}
			\tikzES[-,line width=7pt,line cap=round, line join=round, opacity=0.2,color=green!50!black]{v1/v0,v0/v3}
			\tikzES[-,line width=7pt,line cap=round, line join=round, opacity=0.2,color=blue!50!black]{v0/v4,v1/v2,v2/v3}
		\end{scope}

		\begin{scope}[xshift=1*\xsh*\xr cm]
			\node (e1) at (0*\xr,1*\yr)[xnode]{};
			\node (e2) at (2*\xr,1*\yr)[xnode]{};
			\node (e3) at (2*\xr,0*\yr)[xnode]{};
			\node (e4) at (0*\xr,-1*\yr)[xnode]{};
			\node (e5) at (2*\xr,-1*\yr)[xnode]{};

			\node at (e1) [label={[label distance=-3pt]090:{$a_e^1$}}] {};
			\node at (e2) [label={[label distance=-3pt]090:{$a_e^2$}}] {};
			\node at (e3) [label={[label distance=-2pt]000:{$a_e^3$}}] {};
			\node at (e4) [label={[label distance=-3pt]270:{$a_e^4$}}] {};
			\node at (e5) [label={[label distance=-3pt]270:{$a_e^5$}}] {};

			\tikzES{e1/e3,e3/e4}
			\tikzES[-,thick,dashed,decoration={markings,mark=at position 0.5 with {\arrow{>}}},postaction={decorate}]{e2/e1,e5/e4}
			\tikzES[-,ultra thick, color=red]{e1/e4,e2/e3,e3/e5}
			\tikzES[-,line width=7pt,line cap=round, line join=round, opacity=0.2,color=green!50!black]{e1/e2,e3/e4}
			\tikzES[-,line width=7pt,line cap=round, line join=round, opacity=0.2,color=blue!50!black]{e1/e3,e4/e5}
		\end{scope}

		\begin{scope}[xshift=2*\xsh*\xr cm]

			\node (v0) at (1*\xr,0*\yr)[xnode]{};
			\node (v1) at (0*\xr,1*\yr)[xnode]{};
			\node (v2) at (2*\xr,1*\yr)[xnode]{};
			\node (v3) at (2*\xr,-1*\yr)[xnode]{};
			\node (v4) at (0*\xr,-1*\yr)[xnode]{};

			\node at (v2) [label={[label distance=-3pt]090:{$b_u^2 = a_e^2$}}] {};
			\node at (v3) [label={[label distance=-3pt]270:{$b_u^3 = a_e^1$}}] {};

			\node (e1) at (4*\xr,1*\yr)[xnode]{};
			\node (e2) at (4*\xr,-1*\yr)[xnode]{};
			\node (e3) at (3*\xr,1*\yr)[xnode]{};

			\tikzES{v0/v1,v0/v3,v1/v4,v4/v3,v0/v2};
			\tikzES[-,thick,dashed,decoration={markings,mark=at position 0.5 with {\arrow{>}}},postaction={decorate}]{v0/v4,v2/v1,v2/v3,e1/e2};

				\tikzES{v2/e3,e3/e1,e3/e2,e3/v3,e2/v3};
		\end{scope}
		\end{tikzpicture}
		\caption{Illustrations for the construction for \cref{thm:nph:deg5}(i).
			\emph{Left:} A vertex gadget for a vertex $u$.
			Red edges are size-two habitats (and thus forced).
			Blue edges mark the set $F_u^\top$.
			Green edges mark the set $F_u^\bot$.
			Dashed edges mark the docking pairs.
			\emph{Center:} An edge gadget for an edge $e = \{u,v\}$.
			Green edges mark the set $F_e^u$.
			Blue edges mark the set $F_e^v$.
			\emph{Right:} The result of docking the two gadgets at $(b_u^2, b_u^3)$ and $(a_e^2, a_e^1)$.
		}
		\label{fig:NPh:53}
	\end{figure}
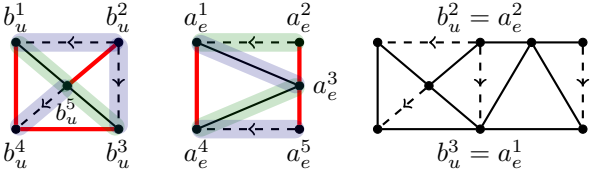

\subsection{Proof of \cref{thm:nph:deg5}\eqref{thm:nph:deg5:hab3}}

\begin{construction}
	\label{constr:nph-deg-5}
	Let $\III = (G = (V, E), k)$ be an instance of \PCVC{} with $V = \{1, \dots, n\}$ and $m \coloneqq |E|$.
	We construct an instance $\III' = (G', \calH, \setFc, k')$ of \TwoDiam{} with $k' \coloneqq 5n + 4m + k$
	by creating a vertex gadget for each $v \in V$ and an edge gadget for each $e \in E$.
	The gadgets will afterwards be connected using the docking operation from \cref{def:nph-docking}
	(see~\cref{fig:NPh:53} for illustration).
	Let $u \in V$. Then the \emph{vertex gadget for $u$} consists of
	the graph $G_u = (V_u, E_u)$ with vertex set $V_u \coloneqq \{b_u^1, \dots, b_u^5\}$,
	and edge set $E_u \coloneqq F^*_u \cup F_u^\top \cup F_u^\bot$,
	where $F^*_u \coloneqq \{ \{b_u^1,b_u^4\},\allowbreak \{b_u^2, b_u^5\},\allowbreak \{b_u^3, b_u^4\} \}$,
	$F^\top_u \coloneqq \{ \{b_u^1, b_u^2\},\allowbreak \{b_u^2, b_u^3\},\allowbreak \{b_u^4, b_u^5\} \}$,
	and $F^\bot_u \coloneqq \{ \{b_u^1, b_u^5\},\allowbreak \{b_u^3, b_u^5\} \}$,
	and the habitat set $\HHH_u \coloneqq F^*_u \cup \allowbreak \{ Q \mid Q \subseteq V_u \text{ induces a triangle} \}$.
	The vertex gadget has three \emph{docking pairs}, $d_u^1 \coloneqq (b_u^2, b_u^1)$, $d_u^2 \coloneqq (b_u^2, b_u^3)$, and $d_u^3 \coloneqq (b_u^5, b_u^4)$.
	See \cref{fig:NPh:53} (left) for an illustration.

	Let $e = \{u,v\} \in E$. Then the \emph{edge gadget for $e$} consists of
	the graph $G_e = (V_e, E_e)$ with vertex set $V_e \coloneqq \{a_e^1, \dots, a_e^5\}$,
	and edge set $E_e \coloneqq F_e^* \cup F_e^u \cup F_e^v$,
	where $F^*_e \coloneqq \{ \{a_e^1, a_e^4\},\allowbreak \{a_e^2, a_e^3\},\allowbreak \{a_e^3, a_e^5\} \}$,
	$F^u_e \coloneqq \{ \{a_e^1, a_e^2\},\allowbreak \{a_e^3, a_e^4\} \}$,
	and $F^v_e \coloneqq \{ \{a_e^1, a_e^3\},\allowbreak \{a_e^4, a_e^5\} \}$,
	and the habitat set $\HHH_e \coloneqq F^*_e \cup \allowbreak \{ Q \mid Q \subseteq V_e \text{ induces a triangle} \}$.
	The edge gadget has two \emph{docking pairs} $d_e^u \coloneqq (a_e^1, a_e^2)$ and $d_e^v \coloneqq (a_e^4, a_e^5)$.
	See \cref{fig:NPh:53} (center) for an illustration.

	We now construct the graph $G'$ by doing the following for each $u \in V$.
	For $i \in \{1, 2, 3\}$ let $e_i = \{u, v_i\}$ be the edges incident to $u$ such that $v_1 < v_2 < v_3$ (for some global linear order $<$, i.e., the numbering of $V$).
	Then, for $i \in \{1, 2, 3\}$,
	we dock the vertex gadget $G_u$ and and the edge gadget $G_{e_i}$ at $d_u^i$ and $d_{e_i}^u$ (after each docking, we refer to the obtained graph still as both $G_u$ and $G_{e_i}$.
	See \cref{fig:NPh:53} (right) for an illustration.

	Finally, we set $\HHH \coloneqq \bigcup_{u \in V} \HHH_u \cup \bigcup_{e \in E} \HHH_e$.
\end{construction}

\begin{observation}
	\label{obs:nph-deg-5}
	Let $\III$ be an instance of \PCVC{} and let $\III'$ be the instance obtained from running \cref{constr:nph-deg-5} on $\III$.
	Then the graph $G'$ in $\III'$ obtained from \cref{constr:nph-deg-5} has maximum degree five, each habitat $H \in \HHH$ has size at most three,
	and obtaining $\III'$ takes polynomial time.
\end{observation}
\begin{proof}
	Habitat size and running time are straightforward.
	Let $u \in V$. %
	Observe that each docking pair is used for a docking operation exactly once.
	Consider the vertices in the docking pairs $(b_u^2, b_u^1), (b_u^2, b_u^3), (b_u^5, b_u^4)$.
	Observe that $\deg_{G_u}(b_u^1) = \deg_{G_u}(b_u^3) =\allowbreak \deg_{G_u}(b_u^4) = 3$, and the vertices are docked to either $a_e^1$ or $a_e^4$, where $e$ is one of the edges incident to $u$.
	As $\deg_{G_e}(a_e^1) = \deg_{G_e}(a_e^4) = 3$, we have that $\deg_{G'}(b_u^1) = \deg_{G'}(b_u^3) = \deg_{G'}(b_u^4) = 5$ in $G'$.
	The vertex $b_u^5$ has degree $\deg_{G_u}(b_u^5) = 4$, and is docked to either $a_e^2$ or $a_e^5$ for some $e$ incident to $u$.
	Thus, in $G'$, we have $\deg_{G'}(b_u^5) = 5$.
	Finally, $\deg_{G_u}(b_u^2) = 3$, but $b_u^2$ is part of two docking pairs.
	As it is docked to either $a_e^2$ or $a_e^5$ for some $e$ incident to $u$
	and to $a_{e'}^2$ or $a_{e'}^5$ for another ${e'}$ incident to $u$,
	we have $\deg_{G'}(b_u^2) = 5$.
	This bounds the degrees of all vertices that are part of a docking operation.
	The remaining vertices in $G'$ are $a_e^3$, $e \in E$, and have degree four.
\end{proof}

For a graph $G = (V, E)$ and a habitat $H \subseteq V$, we call an edge set $F \subseteq E(G)$ a \emph{subsolution} for $H$ if $\diam(G[F][H]) \le 2$.

\begin{observation}
	\label{obs:nph-deg-5-sol}
	Let $\III'$ be an instance from \cref{constr:nph-deg-5}.
	Then for each $u \in V$ and each $H \in \HHH_u$,
	$F_u^* \cup F_u^\top$ and $F_u^* \cup F_u^\bot$ are valid subsolutions for $H$,
	and for each $e = \{u,v\} \in E$ and each $H \in \HHH_e$,
	$F_e^* \cup F_e^u$ and $F_e^* \cup F_e^v$ are valid subsolutions for $H$,
\end{observation}
\begin{proof}
	Let $u \in V$ and let $F_u \in \{ F_u^* \cup F_u^\top, F_u^* \cup F_u^\bot \}$.
	Every habitat in $\HHH_u$ of size two is covered by an edge in $F_u^*$.
	Every other habitat in $H \in \HHH_u$ induces a triangle in $G'$; thus $G'[F_u][H]$ must contain at least two edges to be a subsolution for $H$.
	It is easy to verify that this is true.
	The proof for each $\HHH_e$, $e \in E$ is analogous.
\end{proof}

\begin{lemma}
	\label{lem:nph-deg-5-forward}
	If $\III$ is a \yes-instance, then so is the instance $\III'$ obtained from \cref{constr:nph-deg-5}.
\end{lemma}
\begin{proof}
	Let $S \subseteq V$, $|S| \le k$ be a vertex cover in $G$.
	Then there is a function $g \colon E \to S$ such that $g(\{u, v\}) \in \{u, v\}$.
	Let $F \coloneqq F_V \cup F_E$, where
	$F_V \coloneqq \bigcup_{u \in V} F_u^* \cup \bigcup_{u \notin S} F_u^\bot \cup \bigcup_{u \in S} F_u^\top$, and
	$F_E \coloneqq \bigcup_{e \in E} (F_e^* \cup F_e^{g(e)})$.
	Then, by \cref{obs:nph-deg-5-sol} we have $\diam(G'[F][H]) \le 2$ for every $H \in \HHH$.
	It remainds to show $|F| \le k'$.
	Note that
	\[
		|F_V| \le %
		(n-k) \cdot (3 + 2) + k \cdot (3 + 3) =
		5n + k.
	\]
	As for every $e \in E$, $|F_V \cap F_e^{g(e)}| = |F_{g(e)}^\top \cap F_e^{g(e)}| = 1$,
	\[
		|F| = |F_V| + |F_E| - |F_V \cap F_E| \le 5n + k + 5m - m = k'.\qedhere
	\]
\end{proof}

\begin{lemma}
	\label{lem:nph-deg-5-backward}
	If the instance $\III'$ obtained from running \cref{constr:nph-deg-5} on $\III$ is a \yes-instance, then so is $\III$.
\end{lemma}
\begin{proof}
	Let $F$ be a solution for $\III'$ of size at most $k'$.
	Clearly, $F^* \coloneqq \bigcup_{u \in V} F_u^* \cup \bigcup_{e \in E} F_e^* \subseteq F$ as it consists of all edges contained in size-two habitats.
	As the above sets are pairwise disjoint, we have $|F^*| = 3(m+n)$.
	Let $\HHH' \coloneqq \{ H \in \HHH \mid |H| \ge 3 \}$.
	Note that $\tilde F$ is a subsolution for a habitat in $\HHH'$ if and only if it contains at least two edges of $G'[H]$.
	Observe also that $F^*$ contains exactly one edge of each habitat $H \in \HHH'$; thus there must be at least one edge of $F \setminus F^*$ in each subgraph $G'[H]$.
	We say that a habitat $H \in \HHH'$ is \emph{covered} by an edge $e \in E' \setminus F^*$ if $e \in E(G'[H])$.
	Lastly, let $\HHH'_u \coloneqq \HHH' \cap \HHH_u$ for each $u \in V$ and $\HHH'_e \coloneqq \HHH' \cap \HHH_e$ for each $e \in E$.
	Note that $|\HHH'_u| = 4$ and $|\HHH'_e| = 3$.

	For each $u \in V$ let $F_u \coloneqq (E_u \cap F) \setminus F^*$.
	Moreover, for each $e = \{u, v\} \in E$, let $F_e \coloneqq (E_e \cap F) \setminus (F^* \cup \bigcup_{w \in V} E_w)$.
	Observe that $|F_e| \ge 1$ as each edge gadget contains one habitat $H \in \HHH'_e$ whose edges do not intersect the edges of any vertex gadget.
	Moreover, we may assume that $|F_e| = 1$.
	Otherwise, if $|F_e| \ge 2$, then $(F \setminus F_e) \cup F_e^u$ is also a valid solution of size at most $k'$ for $\III'$:
	The size constraint is fulfilled since $|F_e^u| = 2$.
	Moreover, $F_e^u$ covers a superset of the edges of $F_e$ as it also covers a habitat in $\HHH'_u$.
	As $|F_e^u \setminus E_u| = 1$, we indeed may assume that $|F_e| = |(E_e \cap F) \setminus (F^* \cup \bigcup_{w \in V} E_w)| = 1$ in our solution.
	Thus, $\sum_{e \in E} |F_e| = m$, and we obtain that $\sum_{u \in V} |F_u| \le 2n + k$.

	It is easy to verify that $|F_u| \ge 2$ for each $u \in V$:
	There are four habitats in $\HHH'_u$, and each of them requires to be covered by an edge in $F_u$.
	As every edge in $E_u$ covers at most two of the mentioned habitats, the claim thus follows by the pigeonhole principle.
	Indeed, if $|F_u| = 2$, then $F_u = F_u^\bot$: the two edges in $F_u^\bot$ are the only pair of edges that cover all habitats in $\HHH'_u$.
	Suppose now that $|F_u| \ge 3$.
	Then we may assume that $F_u = F_u^\top$, as $(F \setminus F_u) \cup F_u^\top$ is also a solution of size at most $k'$ for $\III'$:
	The size constraint is fulfilled since $|F_u^\top| = 3$.
	Clearly, the edges in $F_u^\top$ cover each habitat $H \in \HHH'_u$.
	Moreover, each edge in $F_u^\top$ covers one habitat in $\HHH'_e$, whereas any other edge in $E_u$ does not cover any habitats in $\HHH'_e$.
	Thus, $F_u^\top$ covers at least as many habitats as $F_u$.

	In all, we may assume that $F_u \in \{F_u^\top, F_u^\bot\}$.
	Let $S \coloneqq \{ u \in V \mid F_u = F_u^\top \}$.
	As $\sum_{u \in V} |F_u| \le 2n + k$ and $|F_u| \ge 2$, we have that $|S| \le k$.
	Suppose that there exists an edge $e = \{u, v\} \in E$ with $u, v \notin S$.
	The single edge in $F_e$ can cover only two of the three habitats in $\HHH'_e$;
	thus at least one of the two edges in $E_e \cap (E_u \cap E_v)$ must be in $F$.
	But as $F_u = F_u^\bot$ and $F_v = F_v^\bot$, and $F_u^\bot \cap E_e = F_v^\bot \cap E_e = \emptyset$, there is a habitat in $\HHH'_e$ that is not covered; thus $F$ is not a solution, a contradiction.
\end{proof}

\cref{thm:nph:deg5}\eqref{thm:nph:deg5:hab3} now follows directly from \cref{obs:nph-deg-5,lem:nph-deg-5-forward,lem:nph-deg-5-backward}.

The direction constraint may violate planarity.
Thus, for \cref{thm:nph:deg5}\eqref{thm:nph:deg5:planarhab4},
we need two different edge gadgets (see~\cref{fig:NPh:5p4})
such that the docking pairs point in the same direction in one and
in different directions in the other.
\begin{figure}[t]
		\centering
		\begin{tikzpicture}
		\def\xr{0.775}
		\def\yr{0.8}
		\tikzpreamble{}

		\begin{scope}[yshift=-0.5*\xr cm]
			\node (v0) at (0*\xr,1*\yr)[xnode]{};
			\node (v1) at (2*\xr,1*\yr)[xnode]{};
			\node (v2) at (2*\xr,0*\yr)[xnode]{};
			\node (v3) at (2*\xr,-1*\yr)[xnode]{};
			\node (v4) at (0*\xr,-1*\yr)[xnode]{};
			\node (v5) at (0.5*\xr,0*\yr)[xnode]{};

			\node at (v0) [label=090:{$b_u^1$}] {};
			\node at (v1) [label=090:{$b_u^2$}] {};
			\node at (v2) [label={[label distance=-2pt]000:{$b_u^3$}}] {};
			\node at (v3) [label=270:{$b_u^4$}] {};
			\node at (v4) [label=270:{$b_u^5$}] {};
			\node at (v5) [label=180:{$b_u^6$}] {};

			\tikzES{v0/v5,v1/v5,v2/v5,v2/v3,v3/v5,v4/v5}
			\tikzES[-,thick,dashed,decoration={markings,mark=at position 0.5 with {\arrow{>}}},postaction={decorate}]{v1/v0,v1/v2,v3/v4}
			\tikzES[-,ultra thick, color=red]{v5/v1,v5/v3}
			\tikzES[-,line width=7pt,line cap=round, line join=round, opacity=0.2,color=green!50!black]{v0/v5,v2/v5,v4/v5}
			\tikzES[-,line width=7pt,line cap=round, line join=round, opacity=0.2,color=blue!50!black]{v0/v1,v1/v2,v2/v3,v3/v4}
		\end{scope}

		\begin{scope}[xshift=7.6*\xr cm]
			\node (e1) at (0*\xr,1*\yr)[xnode]{};
			\node (e2) at (2*\xr,1*\yr)[xnode]{};
			\node (e3) at (2*\xr,0*\yr)[xnode]{};
			\node (e4) at (0*\xr,-1*\yr)[xnode]{};
			\node (e5) at (2*\xr,-1*\yr)[xnode]{};
			\node (e6) at (0*\xr,-2*\yr)[xnode]{};
			\node (e7) at (2*\xr,-2*\yr)[xnode]{};

			\node at (e1) [label={[label distance=-2pt]180:{$a_e^1$}}] {};
			\node at (e2) [label={[label distance=-2pt]000:{$a_e^2$}}] {};
			\node at (e3) [label={[label distance=-2pt]000:{$a_e^3$}}] {};
			\node at (e4) [label={[label distance=-2pt]180:{$a_e^4$}}] {};
			\node at (e5) [label={[label distance=-2pt]000:{$a_e^5$}}] {};
			\node at (e6) [label={[label distance=-2pt]180:{$a_e^6$}}] {};
			\node at (e7) [label={[label distance=-2pt]000:{$a_e^7$}}] {};

			\tikzES{e1/e3,e3/e4,e4/e5,e5/e6}
			\tikzES[-,thick,dashed,decoration={markings,mark=at position 0.5 with {\arrow{>}}},postaction={decorate}]{e2/e1,e7/e6}
			\tikzES[-,ultra thick, color=red]{e1/e4,e4/e6,e2/e3,e3/e5,e5/e7}
			\tikzES[-,line width=7pt,line cap=round, line join=round, opacity=0.2,color=green!50!black]{e1/e2,e3/e4,e6/e5}
			\tikzES[-,line width=7pt,line cap=round, line join=round, opacity=0.2,color=blue!50!black]{e1/e3,e4/e5,e6/e7}
		\end{scope}

		\begin{scope}[xshift=3.6*\xr cm]
			\node (e1) at (0*\xr,1*\yr)[xnode]{};
			\node (e2) at (2*\xr,1*\yr)[xnode]{};
			\node (e3) at (2*\xr,0*\yr)[xnode]{};
			\node (e4) at (0*\xr,-1*\yr)[xnode]{};
			\node (e5) at (2*\xr,-1*\yr)[xnode]{};
			\node (e6) at (0*\xr,-2*\yr)[xnode]{};
			\node (e7) at (2*\xr,-2*\yr)[xnode]{};

			\node at (e1) [label={[label distance=-2pt]180:{$a_e^1$}}] {};
			\node at (e2) [label={[label distance=-2pt]000:{$a_e^2$}}] {};
			\node at (e3) [label={[label distance=-2pt]000:{$a_e^3$}}] {};
			\node at (e4) [label={[label distance=-2pt]180:{$a_e^4$}}] {};
			\node at (e5) [label={[label distance=-2pt]000:{$a_e^5$}}] {};
			\node at (e6) [label={[label distance=-2pt]180:{$a_e^6$}}] {};
			\node at (e7) [label={[label distance=-2pt]000:{$a_e^7$}}] {};

			\tikzES{e1/e3,e3/e4,e4/e5,e4/e7}
			\tikzES[-,thick,dashed,decoration={markings,mark=at position 0.5 with {\arrow{>}}},postaction={decorate}]{e2/e1,e6/e7}
			\tikzES[-,ultra thick, color=red]{e1/e4,e4/e6,e2/e3,e3/e5,e5/e7}
			\tikzES[-,line width=7pt,line cap=round, line join=round, opacity=0.2,color=green!50!black]{e1/e2,e3/e4,e4/e7}
			\tikzES[-,line width=7pt,line cap=round, line join=round, opacity=0.2,color=blue!50!black]{e1/e3,e4/e5,e6/e7}
		\end{scope}

		\begin{scope}[xshift=6.155*\xr cm, yshift=-3.06*\yr cm]
			\begin{scope}[rotate=-90]
				\node (v0) at (0*\xr,1*\yr)[xnode]{};
				\node (v1) at (2*\xr,1*\yr)[xnode]{};
				\node (v2) at (2*\xr,0*\yr)[xnode]{};
				\node (v3) at (2*\xr,-1*\yr)[xnode]{};
				\node (v4) at (0*\xr,-1*\yr)[xnode]{};
				\node (v5) at (0.5*\xr,0*\yr)[xnode]{};

				\tikzES{v0/v5,v1/v5,v2/v5,v2/v3,v3/v5,v4/v5}
				\tikzES[-,thick,dashed,decoration={markings,mark=at position 0.5 with {\arrow{>}}},postaction={decorate}]{v1/v0,v1/v2,v3/v4}
				\tikzES[-,ultra thick, color=red]{v5/v1,v5/v3}
			\end{scope}
		\end{scope}

		\begin{scope}[xshift=9.251*\xr cm, yshift=-3.06*\yr cm]
			\begin{scope}[rotate=-90]
				\node (e1) at (0*\xr,1*\yr)[xnode]{};
				\node (e2) at (2*\xr,1*\yr)[xnode]{};
				\node (e3) at (2*\xr,0*\yr)[xnode]{};
				\node (e4) at (0*\xr,-1*\yr)[xnode]{};
				\node (e5) at (2*\xr,-1*\yr)[xnode]{};
				\node (e6) at (0*\xr,-2*\yr)[xnode]{};
				\node (e7) at (2*\xr,-2*\yr)[xnode]{};

				\tikzES{e1/e3,e3/e4,e4/e5,e5/e6}
				\tikzES[-,thick,dashed,decoration={markings,mark=at position 0.5 with {\arrow{>}}},postaction={decorate}]{e2/e1,e7/e6}
				\tikzES[-,ultra thick, color=red]{e1/e4,e4/e6,e2/e3,e3/e5,e5/e7}
			\end{scope}
		\end{scope}

		\begin{scope}[xshift=3.06*\xr cm, yshift=-5*\yr cm]
			\begin{scope}[rotate=90]
				\node (e1) at (0*\xr,1*\yr)[xnode]{};
				\node (e2) at (2*\xr,1*\yr)[xnode]{};
				\node (e3) at (2*\xr,0*\yr)[xnode]{};
				\node (e4) at (0*\xr,-1*\yr)[xnode]{};
				\node (e5) at (2*\xr,-1*\yr)[xnode]{};
				\node (e6) at (0*\xr,-2*\yr)[xnode]{};
				\node (e7) at (2*\xr,-2*\yr)[xnode]{};

				\tikzES{e1/e3,e3/e4,e4/e5,e4/e7}
				\tikzES[-,thick,dashed,decoration={markings,mark=at position 0.5 with {\arrow{>}}},postaction={decorate}]{e2/e1,e6/e7}
				\tikzES[-,ultra thick, color=red]{e1/e4,e4/e6,e2/e3,e3/e5,e5/e7}
			\end{scope}
		\end{scope}

		\begin{scope}[xshift=1*\xr cm, yshift=-5*\yr cm]
			\begin{scope}[rotate=90]
				\node (v0) at (0*\xr,1*\yr)[xnode]{};
				\node (v1) at (2*\xr,1*\yr)[xnode]{};
				\node (v2) at (2*\xr,0*\yr)[xnode]{};
				\node (v3) at (2*\xr,-1*\yr)[xnode]{};
				\node (v4) at (0*\xr,-1*\yr)[xnode]{};
				\node (v5) at (0.5*\xr,0*\yr)[xnode]{};

				\tikzES{v0/v5,v1/v5,v2/v5,v2/v3,v3/v5,v4/v5}
				\tikzES[-,thick,dashed,decoration={markings,mark=at position 0.5 with {\arrow{>}}},postaction={decorate}]{v1/v0,v1/v2,v3/v4}
				\tikzES[-,ultra thick, color=red]{v5/v1,v5/v3}
			\end{scope}
		\end{scope}

		\end{tikzpicture}
		\caption{Illustrations for the construction for \cref{thm:nph:deg5}(ii).
			Red edges are forced.
			Green and blue highlighted are two (sub)solutions.
			\emph{Top Left:} A vertex gadget for a vertex $u$.
			\emph{Top Center and Right:} A default edge gadget for an edge $e$.
			\emph{Top Right:} An anti-crossing edge gadget.
			\emph{Bottom:} The result of docking two vertex and two edge gadgets; the right edge gadget is anti-crossing.
		}
		\label{fig:NPh:5p4}
	\end{figure}
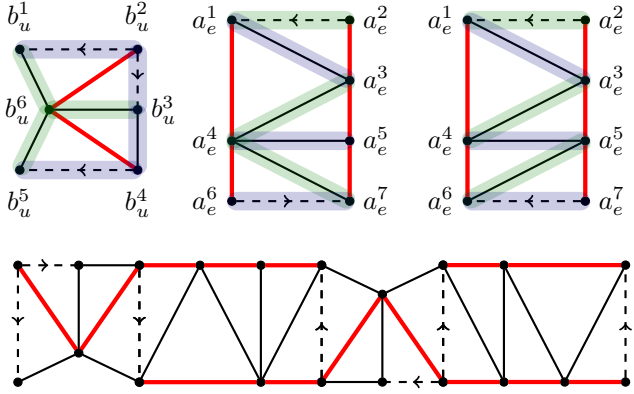
This is only possible with size-four habitats.

\subsection{Proof of \cref{thm:nph:deg5}\eqref{thm:nph:deg5:planarhab4}}

\begin{construction}
	\label{constr:nph-planar}
	Let $\III = (G = (V, E), k)$ be an instance of \PCVC{} with $V = \{1, \dots, n\}$ and $m \coloneqq |E|$.
	We assume to be given a plane embedding of $G$.
	We will construct an instance $\III' = (G', \calH, \setFc, k')$ of \TwoDiam{} with $k' \coloneqq 5n + 7m + k$ as follows
	(see~\cref{fig:NPh:5p4} for an illustration to the gadgets).
	For this, we will create a vertex gadget for each $v \in V$ and an edge gadget for each $e \in E$, which we will afterwards connect using the docking procedure from \cref{def:nph-docking}.

	Let $u \in V$.
	The \emph{vertex gadget} for $u$ consists of the graph $G_u = (V_u, E_u)$ with
	vertex set $V_u \coloneqq \{b_u^1, \dots, b_u^6\}$ and
	edge set $E_u \coloneqq F_u^\top \cup\allowbreak F_u^\bot \cup F_u^*$ with
	$F_u^\top \ceq \bigcup_{i=1}^4 \{\{b_u^i, b_u^{i+1}\}\}$,
	$F_u^\bot \ceq \bigcup_{i\in\{1,3,5\}} \{\{b_u^i, b_u^6\}\}$,
	and $F_u^*\ceq \{\{b_u^2,b_u^6\},\{b_u^4,b_u^6\}\}$,
	and the habitat set
	\[\HHH_u \coloneqq \{ \{b_u^1, b_u^2, b_u^3, b_u^6\}, \{b_u^3, b_u^4, b_u^5, b_u^6\} \}.\]
	The \emph{docking pairs} of the vertex gadget for $u$ are
	$d_u^1 \coloneqq (b_u^2, b_u^1)$,
	$d_u^2 \coloneqq (b_u^2, b_u^3)$, and
	$d_u^3 \coloneqq (b_u^4, b_u^5)$.

	The \emph{default edge gadget} for $e = \{u, v\} \in E$ consists of the graph $G_e = (V_e, E_e)$ with
	vertex set $V_e \coloneqq \{a_e^1, \dots, a_e^7\}$ and
	edge set $E_e \coloneqq F_e^* \cup F_e^u \cup F_e^v$,
	where $F_e^* \coloneqq \{
		\{a_e^1, a_e^4\},\allowbreak
		\{a_e^4, a_e^6\},\allowbreak
		\{a_e^2, a_e^3\},\allowbreak
		\{a_e^3, a_e^5\},\allowbreak
		\{a_e^5, a_e^7\},
	\}$,
	$F_e^u \coloneqq \{
		\{a_e^1, a_e^3\},\allowbreak
		\{a_e^4, a_e^5\},\allowbreak
		\{a_e^6, a_e^7\}
	\}$, and
	$F_e^v \coloneqq \{
		\{a_e^1, a_e^2\},\allowbreak
		\{a_e^3, a_e^4\},\allowbreak
		\{a_e^4, a_e^7\},
	\}$,
	and the habitat set $\HHH_e \coloneqq F^*_e \cup \allowbreak \{ Q \mid Q \subseteq V_e \text{ induces a triangle} \}$.
	The docking pairs of the gadget are
	$d_e^1 \coloneqq (a_e^2, a_e^1)$
	and
	$d_e^2 \coloneqq (a_e^6, a_e^7)$.
	Finally, let~$F_e^d \ceq \{\{a_e^1, a_e^2\},\{a_e^6, a_e^7\}\}$.

	The \emph{anti-crossing edge gadget} for $e = \{u, v\} \in E$
	can be constructed from a default edge gadget for $e$ by taking the graph $G_e$ and replacing the edge $\{a_e^4, a_e^7\}$ in $F_e^u$ with the edge $\{a_e^5, a_e^6\}$.
	The habitat set $\HHH_e$ is constructed as for the default edge gadget.
	The docking pairs of the gadget are
	$d_e^1 \coloneqq (a_e^2, a_e^1)$
	and
	$d_e^2 \coloneqq (a_e^7, a_e^6)$,
	and~$F_e^d \ceq \{\{a_e^1, a_e^2\},\{a_e^6, a_e^7\}\}$.

	Let us now construct the graph $G'$.
	We start with a planar drawing of $G$ and replace each vertex $u \in V$
	with a drawing of a vertex gadget for $u$ such that all vertices in docking pairs lie on the outer face.
	Consider a vertex $u \in V$ and let $e_i = \{u, v_i\}$, $i \in \{1, 2, 3\}$, be the three edges incident to $u$.
	Assume that, when iterating over the edges incident to $u$ in a clockwise order with respect to the drawing of $G$, they appear in the order $e_1, e_2, e_3$.
	Now map $e_i$ to $d_u^i$ for $i \in \{1, 2, 3\}$.
	With the mapping at hand, we can add the edge gadgets.
	Let $e = \{u, v\} \in E$ and let $e_i$ be mapped to $d_u^i$ and to $d_v^j$.
	Replace $e$ with an anti-crossing edge gadget if exactly one of $i$ or $j$ equals $1$,
	and with a default edge gadget otherwise.
	In any case, the drawing of the gadget should be such that the vertices in docking pairs are on the outer face.
	Finally, dock $G_u$ with $G_e$ at $d_u^i$ and $d_e^u$,
	and dock $G_v$ with $G_e$ at $d_v^j$ and $d_e^v$.

	Finally,
	set $\HHH \coloneqq \bigcup_{v \in V} \HHH_v \cup \bigcup_{e \in E} \HHH_e$
	and set $\setFc \coloneqq \bigcup_{v \in V} F_u^* \cup \bigcup_{e \in E} F_e^*$.
\end{construction}

\begin{observation}
	\label{obs:nph-planar}
	Let $\III$ be an instance of \PCVC{} and let $\III'$ be the instance obtained from running \cref{constr:nph-planar} on $\III$.
	Then the graph $G'$ in $\III'$ is planar and has maximum degree five, each habitat $H \in \HHH$ has size at most four,
	and obtaining $\III'$ takes polynomial time.
\end{observation}
\begin{proof}
	Habitat size and running time are straightforward.
	Let us next bound the maximum degree of the vertices in docking pairs; the remaining vertices trivially have degree at most five.
	Observe that each docking pair is used for a docking operation exactly once.
	Consider the vertices in the docking pairs $(b_u^2, b_u^1), (b_u^2, b_u^3), (b_u^5, b_u^4)$.
	Note that the degrees of $b_u^1, b_u^3, b_u^4, b_u^5$ in $G_u$ are at most three,
	and any vertex in a docking pair of an edge gadget $G_e$ has degree at most three in $G_e$.
	Thus, by \cref{obs:nph-docking-deg}, the degree of these vertices is at most five.
	The vertex $b_u^2$ has degree three in $G_u$ and is docked to two edge gadgets.
	It is easy to verify that $b_u^2$ is docked to two vertices that have degree two in their respective edge gadgets.
	Thus, $\deg_{G'}(b_u^2) = 3 + 2 - 1 + 2 - 1 = 5$.

	It remains to prove the planarity of $G'$.
	Observe that in the construction, we follow the clockwise order of the adjacencies for each vertex.
	Moreover, the directions of the docking pairs $d_v^2$ and $d_v^3$ are also clockwise for each $v \in V$,
	but the direction of the docking pair $d_v^1$ is counterclockwise.
	Thus, to ensure that an edge gadget for an edge $e_i = \{u,v\}$ mapped to $d_u^i$ and $d_v^j$ can be drawn without crossing edges (think of replacing each edge with two parallel edges), it must have
	\begin{compactenum}[(i)]
	\item two counterclockwise docking pairs if $i \ne 1 \ne j$,
	\item a clockwise docking pair if either $i=1$ or $j=1$, and
	\item two clockwise docking pairs if $i=j=1$.
	\end{compactenum}
	Note that the latter can be obtained by mirroring the default crossing gadget.
\end{proof}

\begin{lemma}
 \label{lem:deg5:planarhab4:iff}
	Instance $\III$ is a \yes-instance of \PCVC{}
	if and only if
	instance $\III'$ obtained from $\III$ using \cref{constr:nph-planar} is a \yes-instance.
\end{lemma}

\begin{proof}
 \RD{}
 Let~$S\subseteq V$, $|S|=k$ be a vertex cover.
 For~$u\in V$,
 let
 \[
  F_u = F_u^* \cup \begin{cases}
					F_u^\top, &\cif{} u\in S,\\
          F_u^\bot, &\otw.
        \end{cases}
 \]
 Consider an edge~$e\in E$.
 Since~$S$ is a vertex cover,
 one endpoint
 of~$e$,
 say~$u$,
 is contained in~$S$ and hence
 at least one docking edge is contained in~$F_u$.
 Thus,
 let~$F_e\ceq F_e^u\cup F_e^*$.
 We claim that~$F=\bigcup_{v \in V} F_u \cup \bigcup_{e \in E} F_e$ is a solution to~$\III'$.
 Note that~$F^*\subseteq F$.
 Moreover,
 since~$F^*$ is the union of each gadget's solution,
 for each habitat~$H\in\calH$ it holds that~$\diam(G[F][H])\leq 2$.
 Finally,
 \begin{align*}
		& |F|
		=\textstyle \big|\bigcup_{u\in S} F_u \cup \bigcup_{u\in \setminus S} F_u \cup \bigcup_{\{u,w\}\in E} (F_e\setminus (F_u\cup F_w))\big|
		\\
		&=\textstyle \big|\bigcup_{u\in S} F_u\big| + \big|\bigcup_{u\in \setminus S} F_u\big| + \big|\bigcup_{\{u,w\}\in E} (F_e\setminus (F_u\cup F_w))\big|
		\\
		&\leq k\cdot 6 + (n-k)\cdot 5 + m\cdot 7 = k'
 \end{align*}

 \LD{}
 Let~$F$ be a solution to~$\III'$.
 Let~$S' \ceq \{v\in V\mid F\cap F_u^\top \neq \emptyset\}$.
 Let~$S''$ be initially empty
 and let~$E^\emptyset \ceq \{e\in E\mid F\cap F_e^d = \emptyset\}$.
 For each edge~$e\in E^\emptyset$,
 arbitrarily add one endpoint of~$e$ to~$S''$.
 We claim that~$S\ceq S'\cup S''$ is a solution to~$\III$.
 Suppose that there exists an edge $\{u,v\} \in E$ with $u, v \notin S$.
 Then $\{u, v\} \notin E^\emptyset$.
 As the two edges in $F_e^d$ are identified with one edge in $F_u^\top$ and one edge in $F_v^\top$, we have $\{u, v\} \cap S \ne \emptyset$, and $S$ is a vertex cover.
 As~$|F\cap\bigcup_{u\in V} E_u|\geq 6|S'|+5\cdot (n-|S'|) = 5n+|S'|$
 and~$|F\cap \bigcup_{e\in E} (E_e\setminus F_e^d)| \geq 8|E^\emptyset| + 7(m-|E^\emptyset|) = 7m+|E^\emptyset| \geq 7m + |S''|$,
 we get
 \begin{align*}
	 k'\geq |F| &= \textstyle \big|F\cap\bigcup_{u\in V} E_u\big| + \big|F\cap \bigcup_{e\in E} (E_e\setminus F_e^d)\big|\\
		    & \geq 5n + |S'| + 7m + |S''|,
 \end{align*}
 implying that~$|S|\leq |S'|+|S''|\leq k$.
\end{proof}

\section{Epilogue}

We showed that reconnecting habitats so that every habitat has diameter two is computationally challenging
even in sparse graphs with small maximum habitat sizes.
On the way, we also identified intriguing structural properties to give efficient algorithms.

The main unanswered question in this work is whether \TwoDiamC{} is po\-ly\-no\-mi\-al-time solvable on graphs with maximum degree four and maximum habitat size at least five.
We believe that it should be possible to extend our habitat intersection graph approach.
If there are habitats of size six, we can however find another proper docking structure that can also result in paths.
We further remark that the dynamic programming approach for paths in the habitat intersection graph can also be generalized to trees and to graphs of bounded treewidth.
This raises the question: How do instances with these habitat intersection structures look?

Finally, a restriction of our model is that every habitat must have diameter $d=2$.
Of course, one could extend our study by considering the cases $d \ge 3$, or even considering an adaptive diameter bound, i.e., $\diam(G[F][H]) \le \alpha\diam(G[H])$ for $\alpha \ge 1$.

{\begingroup
	\let\clearpage\relax
	\renewcommand{\url}[1]{\href{#1}{$\ExternalLink$}}
	\bibliography{references}
\endgroup}

\end{document}